\documentclass[11pt,letterpaper]{article}
\usepackage[margin=1in]{geometry}
\usepackage[utf8]{inputenc}
\usepackage{amsmath,amssymb,amsthm,mathrsfs}
\usepackage{physics}
\usepackage{complexity}
\usepackage{float}
\usepackage[ruled,linesnumbered,noend]{algorithm2e}
\usepackage{setspace}
\usepackage[noend]{algpseudocode}
\usepackage[colorlinks=true]{hyperref}
\usepackage[dvipsnames]{xcolor}
\definecolor{darkred}  {rgb}{0.5,0,0}
\definecolor{darkblue} {rgb}{0,0,0.5}
\definecolor{darkgreen}{rgb}{0,0.5,0}
\hypersetup{
  urlcolor   = blue,         
  linkcolor  = darkblue,     
  citecolor  = darkgreen,    
  filecolor   = darkred       
}
\usepackage[nameinlink, capitalize]{cleveref}
\usepackage{tcolorbox}
\usepackage{subcaption}
\usepackage{tikz}
\usetikzlibrary{decorations.pathreplacing}
\pgfmathsetmacro\MathAxis{height("$\vcenter{}$")}
\newcommand{\tikzlength}{0.7}
\newcommand{\scale}{0.78}
\newcommand{\mc}{\mathcal}
\renewcommand{\E}{\mathop{\mathbb E\/}}
\renewcommand{\polylog}{\mathrm{polylog}}
\renewcommand{\poly}{\mathrm{poly}}
\usepackage[backend=biber,style=alphabetic,url=false,isbn=false,maxnames=10,minnames=3,maxalphanames=4,minalphanames=3]{biblatex}
\newcommand{\noti}{\mathsf{NI}}
\newcommand{\cnot}{\mathsf{CNOT}}
\newcommand{\ber}{\mathsf{Bern}}
\newcommand{\maj}{\mathsf{Maj}}
\renewcommand{\epsilon}{\varepsilon}
\newcommand{\ep}{\mathsf{EP}}
\newcommand{\under}[2]{\underbrace{#1}_{\substack{#2}}}
\newtheorem{theorem}{Theorem}[section]
\newtheorem{lemma}[theorem]{Lemma}
\newtheorem{question}{Question}
\newtheorem{claim}[theorem]{Claim}
\newtheorem{remark}{Remark}[section]
\newtheorem{definition}{Definition}[section]
\newtheorem{fact}{Fact}[section]
\newtheorem{conjecture}{Conjecture}
\newtheorem{corollary}[theorem]{Corollary}
\SetKwInput{KwInput}{Input}
\SetKwInput{KwOutput}{Output}
\usepackage{algpseudocode}

\begin{document}
\title{Quantum computational advantage with constant-temperature Gibbs sampling}

\author{Thiago Bergamaschi\thanks{UC Berkeley. \href{mailto:thiagob@berkeley.edu}{thiagob@berkeley.edu}}
\and
Chi-Fang Chen\thanks{Caltech. \href{mailto:chifang@caltech.edu}{chifang@caltech.edu}}
\and
Yunchao Liu\thanks{UC Berkeley. \href{mailto:yunchaoliu@berkeley.edu}{yunchaoliu@berkeley.edu}}
}

\date{}

\maketitle

\begin{abstract}

A quantum system coupled to a bath at some fixed, finite temperature converges to its Gibbs state. This \emph{thermalization} process defines a natural, physically-motivated model of quantum computation. However, whether quantum computational advantage can be achieved within this realistic physical setup has remained open, due to the challenge of finding systems that thermalize quickly, but are classically intractable. Here we consider sampling from the measurement outcome distribution of quantum Gibbs states at constant temperatures, and prove that this task demonstrates quantum computational advantage. We design a family of commuting local Hamiltonians (parent Hamiltonians of shallow quantum circuits) and prove that they rapidly converge to their Gibbs states under the standard physical model of thermalization (as a continuous-time quantum Markov chain). On the other hand, we show that no polynomial time classical algorithm can sample from the measurement outcome distribution by reducing to the classical hardness of sampling from noiseless shallow quantum circuits. The key step in the reduction is constructing a fault-tolerance scheme for shallow IQP circuits against input noise.


\end{abstract}

\section{Introduction}
\label{section:introduction}


A major goal of today's quantum computing efforts is to realize quantum computational advantage in realistic physical setups. One such setup is open system thermalization, where a quantum many-body system is specified by a Hamiltonian $H$ and then coupled to a bath at finite (constant) temperature $\beta$, and the system converges to the Gibbs state $\rho_\beta\propto e^{-\beta H}$. Under physical assumptions,\footnote{The bath is Markovian and the coupling is weak.} this thermalization process can be described by a \textit{thermal Lindbladian} (a continuous-time quantum Markov chain), most notably the Davies generator~\cite{Davies1974MarkovianME} and its variants (e.g.~\cite{mozgunov2020completely}). This setup is especially relevant for physical platforms in which implementing digital quantum circuits is difficult. However, there has been no complexity-theoretic evidence showing that quantum computational advantage can be achieved in this model (see~\cref{sec:relatedwork} for a discussion).

In this paper, we provide such evidence by showing that quantum computational advantage can be achieved for the task of sampling from the measurement outcome distribution of Gibbs states at constant temperatures. In particular, we construct a family of commuting local Hamiltonians and show that its thermalization process (described by the Davies generator) is rapidly mixing. Meanwhile, its Gibbs state is classically intractable to sample from.

\begin{theorem}[Main result]
    \label{theorem:main} For any constant inverse-temperature $\beta = \Theta(1)$, there exists a family of $n$-qubit commuting $O(1)$-local Hamiltonians, such that the $n$-qubit Gibbs state $\rho_\beta$ is both
    \begin{enumerate}
        \item \emph{Rapidly Thermalizing.} It can be prepared within small trace distance by the Davies generator (a quantum Markov chain describing thermalization), in time $n^{o(1)}$. In addition, this process can be simulated on a quantum computer in time $n^{1+o(1)}$. And yet, 
        \item \emph{Classically Intractable.} Under certain complexity-theoretic assumptions, there is no polynomial time classical algorithm to sample from the measurement outcome distribution  $p(x)=\bra{x}\rho_\beta\ket{x}$ within small total variation distance.
    \end{enumerate}
\end{theorem}

The classical hardness is based on the hardness of approximate sampling from the output distribution of ideal shallow quantum circuits. The main result, therefore, places the hardness of rapidly mixing thermalization to the same level as ideal sampling-based quantum supremacy experiments (see \cref{section:appendix-hardness} for more details). 

A more general version of the result (\cref{thm:maingeneral}) is given in \cref{section:together}, where we generalize the above and show how to trade-off locality for mixing time, including a family of $O(\log\log n)$-local Hamiltonians which thermalizes in $\polylog(n)$ time.\footnote{In the initial posting of this work, we stated this latter construction as our main result. We thank James Watson and Joel Rajakumar for the observation that under appropriate parameter choices, our construction in fact has constant locality (see \cref{section:overview-distillation} and \cite{rajakumar2024gibbs}).}

\begin{figure}[t]
    \centering
    \begin{subfigure}[b]{0.4\textwidth}
    \begin{tikzpicture}[baseline={(0, 0.5*\tikzlength cm-\MathAxis pt)},x=\tikzlength cm,y=\tikzlength cm]
  \fill[blue!30!white,rounded corners=5pt]  (2.7,3) -- (4,4.3) -- (5.3,3) 
 -- (4,1.7) -- cycle;
 
  \draw[step=1,black] (0,0) grid (5,5);

  \foreach \x in {0,...,5}
  \foreach \y in {0,...,5}
  {
  \filldraw[black] (\x,\y) circle (2pt);}
   \node at (4,3)[anchor=north west]{$i$};
   \node at (5.2,3)[anchor=west]{$h_i=-C Z_i C^\dag$};
\end{tikzpicture}
    \caption{A Local Hamiltonian}
    \label{fig:localhamiltonian}
    \end{subfigure}
    \quad
    \begin{subfigure}[b]{0.4\textwidth}
    \begin{equation*}
    \rho_\beta\quad=\quad\begin{tikzpicture}[baseline={(0, 0.5*\tikzlength cm-\MathAxis pt)},x=\tikzlength cm,y=\tikzlength cm]
        \draw[black] (0,0) rectangle node{$C$} (6,1);
  \foreach \x in {1,2,...,6} {
        \draw[black] (\x-0.5, 0) -- (\x-0.5,-0.8);
        \draw[black] (\x-0.5, 1) -- (\x-0.5,1.5);
        
        \node at (\x-0.5,-0.4)[circle,fill=blue!60,inner sep=1.5pt]{};
        \node at (\x-0.5,-0.8)[anchor=north]{$\ket{0}$};
    }
    \end{tikzpicture}
\end{equation*}
    \caption{The Noise Model}
    \label{fig:noisemodel}
    \end{subfigure}
    \caption{(a) We consider local Hamiltonians $H=\sum_i h_i$ which are parent Hamiltonians of shallow quantum circuits. (b) The Gibbs states of these Hamiltonians $\rho_\beta\propto e^{-\beta H}$ are equivalent to the output state of $C$, where the input qubits are subject to bit-flip errors (blue dots) of rate $(1+e^{2\beta})^{-1}$.} 
    \label{fig:introduction}
\end{figure}

\paragraph{Our approach.} The family of Hamiltonians we consider is the class of ``parent'' Hamiltonians of shallow quantum circuits (\cref{fig:localhamiltonian}). Starting from a trivial, non-interacting Hamiltonian $H_\noti = -\sum_{i}Z_i$ consisting of single-qubit Pauli-$Z$ terms, we consider the family of Hamiltonians that are related to $H_\noti $ by a low depth circuit, 
\begin{equation}\label{eq:parenthamiltonian}
    \mathscr{H}=\left\{H:\exists\text{ low-depth circuit }C,\text{ }H=C H_\noti C^\dag\right\}.
\end{equation}
Each $H\in\mathscr{H}$ is local, commuting, and it encodes the computation $C$ in the sense that its ground state is the output of the circuit $C\ket{0^n}$. The reason that these Hamiltonians are good candidates for quantum advantage at constant temperatures lies in the following key observation:

\begin{quote}
\centering
\emph{
    The Gibbs state of each $H\in\mathscr{H}$ is a noisy version of the underlying computation, where random bit-flip errors are applied to the input qubits (\cref{fig:noisemodel}).
}
\end{quote}

This is a clean example of the general intuition that constant-temperature Gibbs states are very noisy and far from ground states. To encode computational hardness into the Gibbs states of $H\in\mathscr{H}$, it then suffices to design a shallow quantum circuit which is classically intractable to simulate even under input noise. Our main result then follows from two key technical contributions:

\begin{enumerate}
    \item \textbf{A construction of classically-hard shallow quantum circuits that are fault-tolerant against input noise.} Standard techniques in quantum fault-tolerance blow up the circuit depth, and in turn, the locality of the parent Hamiltonian\footnote{Our interest in decreasing the locality stems both from the practical challenges behind engineering systems with many-body interactions, and a complexity-theoretic understanding of the role of locality in the hardness of Gibbs sampling.}. We start from a specific family of classically-hard shallow circuits (namely, IQP circuits~\cite{Gao2016QuantumSF,BermejoVega2017ArchitecturesFQ}), and then design a low-overhead fault-tolerance scheme tailored to IQP circuits and the input noise model.
    
    \item \textbf{A proof that these Hamiltonians thermalize rapidly, via a modified log-Sobolev inequality.} We prove a rapid mixing bound for Hamiltonians in $\mathscr{H}$ which leverages the structure of the thermal Lindbladian (the quantum Markov chain describing thermalization), in combination with a carefully constructed lightcone argument for shallow quantum circuits.
\end{enumerate}


\subsection{Related work}
\label{sec:relatedwork}
\paragraph{Complexity of Gibbs states.} Establishing quantum computational advantage with constant-temperature Gibbs sampling faces inherent difficulties. After all, at high enough temperatures, Gibbs states are expected to be essentially classical objects; in particular, sampling from these Gibbs states is efficient to simulate on a classical computer\footnote{This does not contradict our result which holds for arbitrary constant temperature, due to the order of quantifiers; see~\cref{remark:order-of-quantifiers}.}~\cite{yin2023polynomialtime,bakshi2024hightemperature}. On the other hand, in the low temperature regime, preparing Gibbs states is expected to be hard in general even for a quantum computer;\footnote{Indeed, $\NP$-hard due to the classical PCP theorem \cite{AroraLMSS98}.} in particular, the thermalization process may take exponential time. 


Nevertheless, a path exists to circumvent these issues, by embedding some classically hard quantum computation into a local Hamiltonian. It is reasonable to hope that the nature of this embedding ensures that producing the Gibbs state is still tractable for quantum computers\footnote{In fact, we desire something even stronger: that the Hamiltonian is rapidly thermalizing.} (e.g.~\cite{Aharonov2007Adiabatic,chen2023local}), and one can further hope that the Gibbs state is classically hard. But there is yet another issue: standard means to embed quantum circuits into Hamiltonians \cite{kitaev02} typically encode the quantum computation into its ground state. However, Gibbs states at constant temperatures are understood to be very noisy, and far from the ground state. In this manner, to argue that this noisy version of the ground state remains classically hard, there must be an inherent \emph{fault-tolerance} to the circuit-to-Hamiltonian mapping. Our approach can be viewed as a clean example that satisfies all of the above criteria.

\paragraph{Gibbs samplers and rapid mixing.} Preparing Gibbs states (or Gibbs sampling) is a candidate application of quantum computers as well as an important quantum algorithmic primitive. While there are many proposed quantum Gibbs samplers, recent developments have focused on an approach of simulating open system (Lindbladian) dynamics, in particular the Davies generator and its variants which mimic thermalization in nature~\cite{Rall2023thermalstate,Chen2023QuantumTS,Chen2023AnEA}. 

The key missing ingredient to the efficiency of these quantum simulation algorithms is a bound on the mixing time of the underlying quantum Markov chain. The standard approach, via a bound on the spectral gap, gives a mixing time that has intrinsic polynomial dependence in $n$~\cite{Kastoryano2012QuantumLS}. A much stronger approach known as (quantum) \textit{log-Sobolev inequalities} consists of a decay of the relative entropy, and results in only $\polylog(n)$ mixing time, a phenomenon known as \textit{rapid mixing}. These stronger inequalities are notoriously hard to prove: examples have only been shown for certain commuting systems, in 1D~\cite{Bardet2023Rapid,Bardet2024} or on lattices above a threshold temperature~\cite{capel2021modified}. Our rapid mixing bound uses the lightcone structure of shallow quantum circuits, and does not require geometric locality or a temperature threshold.

\paragraph{Shallow quantum circuits and fault-tolerance.} 

Shallow quantum circuits are widely used in quantum algorithms for near-term devices and quantum supremacy experiments. The hardness of sampling from the output distribution of shallow quantum circuits provides the complexity foundation for these experiments (see~\cite{Hangleiter2022ComputationalAO} for a review). We focus on constant-depth \textit{instantaneous quantum polynomial time} (IQP) circuits $C = H^{\otimes n} D H^{\otimes n}$ where $D$ is a constant-depth diagonal unitary, which provides hardness due to the universality of measurement-based quantum computation~\cite{Gao2016QuantumSF,BermejoVega2017ArchitecturesFQ}. However, these circuits are not noise-robust and become classically simulable under noise~\cite{Bremner2016AchievingQS, Rajakumar2024PolynomialTimeCS}; fault-tolerance techniques are therefore necessary for classical hardness in our context.

There is a tension between shallow quantum circuits and the overhead of quantum fault-tolerance.\footnote{Note that some models of fault-tolerance assume instant classical computation and feedforward within a quantum circuit \cite{mezher20,Paletta2023RobustSI}. This is not allowed in our setting: all operations must be realized by quantum gates.} Standard techniques encode a constant depth quantum circuit into a fault-tolerant circuit of $\polylog(n)$ depth~\cite{Aharonov-Ben-Or}, and fault-tolerance with constant circuit depth overhead is only known for shallow Clifford circuits~\cite{Bravyi2020Quantum}. Ref.~\cite{Bremner2016AchievingQS} devised a fault-tolerance scheme specialized to IQP circuits and the input noise model, and we design a new scheme in this setting which achieves a significantly smaller overhead.

\subsection{Our Contributions}


\subsubsection{Efficient quantum Gibbs sampling via rapid mixing}
\label{subsection:contributions-samplers}

Our first result is a quantum algorithm for preparing the Gibbs states of $H\in\mathscr{H}$, given only a description of its local terms $H =\sum_i h_i$ (as $2^\ell\times 2^\ell$ matrices).\footnote{Although $H$ has a simple structure by definition, the underlying global structure (the low-depth circuit $C$) is hidden among the local terms, and is not directly accessible. See \cref{remark:gate-obfuscation} for a discussion.}

\begin{lemma}[Gibbs State Preparation]\label{lemma:results-gibbs-prep}
    Fix $\beta>0$, and let $H\in\mathscr{H}$ be the parent Hamiltonian of a quantum circuit on $n$ qubits, of depth $d$ and lightcone size $\ell$. Then, there exists a quantum algorithm which can prepare the Gibbs state of $H$ at inverse-temperature $\beta$ up to an error $\epsilon$ in trace distance in time $O(2^{4\ell}\cdot 2^d\cdot  e^{2\beta} \cdot n \cdot \poly (\log \frac{n}{\epsilon}, \ell, \beta))$.
\end{lemma}

\noindent In general, the lightcone size $\ell$ is upper bounded by $\ell\leq 2^d$. We emphasize we do not make any assumptions on the temperature or geometric locality. This is important as our fault-tolerant circuits (\cref{lemma:results-iqp-ft}) are not naturally defined on a lattice. 

The algorithm in \cref{lemma:results-gibbs-prep} follows from a two-step argument. The first step is the design and analysis of a particular family of Davies generators \cite{Davies1974MarkovianME}, a family of dissipative Lindbladians whose local jumps (or transitions) are engineered to resemble the connectivity of the Hamiltonian. In \cref{lemma:mixingtimebound}, we prove that the mixing time of our Lindbladians is $t_{mix}=O(4^\ell\log n)$ via a modified log-Sobolev inequality. In principle, this step is already a \textit{thermal algorithm}, in the sense that ``placing the system in a fridge'' would drive it to the Gibbs state in time $t_{mix}\cdot \log (1/\varepsilon)$.

The second step is the simulation of the dissipative (non-unitary) dynamics on a quantum computer. We employ the block-encoding framework of \cite{Chen2023QuantumTS} which we significantly simplify as our family of Hamiltonians is commuting and has integer spectra. The quantum simulation adds a factor of $n$ to the running time, which may be hard to improve due to the absence of geometric locality. In \cref{sec:latticeprep}, we discuss an alternative method for Gibbs state preparation assuming finite dimensional lattice geometry using the framework of~\cite{Brando2016FiniteCL}.

\subsubsection{Fault-tolerance of shallow IQP circuits}
\label{subsection:contributions-ft}
The key ingredient for the classical hardness of sampling from quantum Gibbs states is to produce a shallow quantum circuit which is hard to sample from even under input noise. For this purpose, we design a fault-tolerance scheme for shallow IQP circuits~\cite{Gao2016QuantumSF,BermejoVega2017ArchitecturesFQ} since their gate set works nicely with fault tolerance techniques. Our result ensures that any IQP circuit can be made robust to input noise with only a small additive blow-up to the circuit depth (see \cref{lemma:IQP-FT-Distillation} for a more general statement). 

\begin{lemma}\label{lemma:results-iqp-ft}
Let $p<\frac{1}{2}$ be a constant bit-flip error rate, and let $C$ be an $n$ qubit IQP circuit of depth $d$. Then, there exists an $O(n\log \frac{n}{\epsilon})$ qubit circuit $\Tilde{C}$ of depth $d + o(\log \frac{n}{\epsilon})$, such that a sample from $\Tilde{C}$ under input noise (\cref{fig:noisemodel}) can be efficiently post-processed into a sample within $\epsilon$ total variation distance to the output distribution of $C$. 
\end{lemma}

This result significantly reduces the blow-up in circuit depth compared to a prior fault-tolerance scheme of~\cite{Bremner2016AchievingQS}. Moreover, the locality of the resulting parent Hamiltonian $H=-\sum_i \Tilde{C} Z_i \Tilde{C}^\dag$ is only a constant. Our key idea is a non-adaptive state distillation scheme, drawing inspiration from magic state distillation~\cite{Bravyi2005Universal}: distilling a near-perfect initial state from noisy initial states, up to a known but uncorrected Pauli error. The error is propagated through the circuit and corrected in post-processing, similar to~\cite{Bravyi2020Quantum}. Propagating Pauli errors through non-Clifford circuits is hard in general, but here it works thanks to the structure of IQP circuits.

\subsubsection{Applications}

\paragraph{BQP Completeness under adaptive single-qubit measurements.} In addition to quantum advantage, using our techniques we can show that constant-temperature Gibbs states do have some inherent form of universality for quantum computation. In \cref{section:completeness} we prove that there exist local Hamiltonians whose Gibbs states are universal resource states for quantum computation, in the sense that they can be used for universal measurement-based quantum computation. 

\begin{theorem}\label{theorem:results-mbqc}
    Fix an inverse-temperature $\beta = \Theta(1)$. Then, there exists an $n$-qubit,  $O(1)$-local commuting Hamiltonian, whose Gibbs state at inverse-temperature $\beta$ is a universal resource state for quantum computation and is efficiently preparable on a quantum computer. 
\end{theorem}

\cref{theorem:results-mbqc} is based on the universality of cluster-states for measurement-based quantum computation. That is to say, any quantum computation of bounded size can be implemented using adaptive single-qubit measurements on top of a fixed 2D cluster-state (e.g.~\cite{Broadbent2008UniversalBQ}). We design a Hamiltonian whose Gibbs state resembles a noisy version of a cluster-state, such that under adaptive single-qubit operations, one can nevertheless correct and distill out computation. 

\paragraph{Gibbs sampling under measurement errors.} An interesting question is whether the thermal quantum advantage demonstrated in this paper, is itself robust to noise. That is, in realistic physical platforms, we expect imperfect state preparation, noisy system-bath couplings, and erroneous measurements. As a starting point to this problem, we consider a model where the Gibbs state preparation is ideal, but there are random bit-flip errors in the measurement outcome.

We show that the quantum advantage survives in this model, albeit at a higher locality. 

\begin{theorem}\label{theorem:results-measurement-errors}
    Fix an inverse temperature $\beta = \Theta(1)$, and a measurement error rate $p < \frac{1}{2}$. There exists a family of $n$-qubit, $O(\log n)$-local Hamiltonians, such that sampling from their Gibbs state at inverse-temperature $\beta$, under measurement errors of rate $p$, is classically intractable under certain complexity-theoretic assumptions. Moreover, there exists a $\poly(n)$ time quantum algorithm to produce said Gibbs state.
\end{theorem}

The Hamiltonians of \cref{theorem:results-measurement-errors} are similar to that of \cref{theorem:main}, in the sense that they are parent Hamiltonians of fault-tolerant IQP circuits. However, to ensure classical hardness under measurement errors, our quantum circuits now need to be fault-tolerant against both input and output errors. (Recall that the ``input errors'' come from temperature, while ``output errors'' come from actual physical noise in measurements.) To do so, in \cref{sec:measurementnoise} we appeal to an optimized construction of a prior fault-tolerance scheme by \cite{Bremner2016AchievingQS}, at the cost of an increase to the locality of the Hamiltonians, which also changes the mixing time from $n^{o(1)}$ to $\poly(n)$.

\subsection{Discussion}
\label{section:discussion}

We conclude by discussing two future directions, broadly related to the complexity of Gibbs sampling. The first of which concerns the BQP Completeness of Gibbs sampling (without adaptivity). 

\begin{question}[BQP Completeness of Gibbs Sampling]
For every $n$ qubit, $\poly(n)$ depth quantum circuit $C$, does there exist a Hamiltonian $H$ and a constant inverse-temperature $\beta>0$ such that by sampling from its Gibbs state one can recover the output of the quantum computation $C$?
\end{question}

Partial progress on this question has recently been made by~\cite{chen2023local}, albeit, only at very low temperatures where the Gibbs state approximates the ground state. In particular, they showed how to embed an arbitrary quantum computation into a (modified) Feynman-Kitaev circuit-to-Hamiltonian mapping, which could be efficiently prepared by a Lindbladian evolution. Whether similar ideas could work at constant temperatures remains an open problem.

Another interesting direction lies in the time overhead for fault-tolerance, and for quantum advantage using shallow circuits which are robust to noise. 

\begin{question}[Quantum Advantage in Noisy Shallow Circuits] 
Does there exist a family of constant depth quantum circuits (using only quantum gates) which is classically hard to sample from in the presence of depolarizing noise on each gate?
\end{question}

Its main motivation lies in the design of quantum advantage experiments, which can be implemented on near-term devices. Depolarizing noise on each gate of the circuit, however, is naturally a significantly more general noise model than input noise. Nevertheless, the same question with input noise remains open as well.

\section{Technical Overview}
In this section we give a sketch of our two main technical contributions:  (1) A proof of a modified log-Sobolev inequality for a family of Davies generators, via a lightcone argument (\cref{section:gibbs_state_prep}); and (2) A fault-tolerance scheme for shallow IQP circuits against input noise, via non-adaptive state distillation (\cref{section:hardness}). We begin by presenting some basic notation and background on thermal Lindbladians.

\subsection{Gibbs state preparation via rapid mixing}
\label{section:gibbs_state_prep}



Fix a Hamiltonian $H\in \mathscr{H}$. By definition, there exists a shallow circuit $C$ such that

\begin{equation}\label{eq:defhamiltonianoverview}
    H = \sum_{i\in [n]} h_i,\quad\text{where}\quad h_i = C\left(\ketbra{1}_i\otimes \mathbb{I}_{[n]\setminus i}\right) C^\dagger,
\end{equation}
and each $\ketbra{1}_i$ is a single-qubit projection. Note that \cref{eq:defhamiltonianoverview} is equivalent to \cref{eq:parenthamiltonian} up to a shift. The eigenstates of $H$ of energy $k\in [n]$ are all the states $C\ket{x}$, where $x\in \{0, 1\}^n$ has Hamming weight $|x| = k$. We denote the projection $\Pi_k$ onto the eigenspace of $H$ of energy $k$ as
\begin{equation}\label{eq:projeigenspace}
    \Pi_k = C\bigg(\sum_{|x| = k}\ketbra{x} \bigg) C^\dagger.
\end{equation}

\noindent We consider two notions of locality for $C$ and $H$ respectively:
\begin{itemize}
    \item The circuit lightcone. The \emph{lightcone} $\mathsf{L}_i$ of qubit $i$ is the set of qubits that can be reached by $i$ via gates in $C$, and we define the \emph{lightcone size} as $\ell=\max_i|\mathsf{L}_i|$.
    \item The Hamiltonian locality. Let $\mathsf{S}_i=\mathrm{supp}(h_i)=\mathrm{supp}(C Z_i C^\dag)$ be the set of qubits that $h_i$ acts nontrivially on. The \emph{locality} of the Hamiltonian $H$ is defined as $r=\max_i |\mathsf{S}_i|$. 
\end{itemize}

Note that $\mathsf{S}_i$ is related to the propagation of $Z_i$ under $C$, and thus by definition we have $\mathsf{S}_i\subseteq \mathsf{L}_i$. In fact, $r\ll \ell$ for the family of circuits we consider.

\subsubsection{Our Davies generators}
We determine our family of thermal Lindbladians, or Davies generators, by specifying two ingredients: a set of jump operators, and transition weights. Technically, general thermal Lindbladians for noncommuting Hamiltonians need not take the Davies' form (see, e.g,~\cite{Chen2023QuantumTS,Chen2023AnEA}), but for commuting Hamiltonians, the Davies' generator is nonetheless sufficient for all our discussions.

\begin{itemize}
    \item \textbf{Jump Operators.} To generate the transitions, we consider the set of jump operators which are local, $\ell$-qubit Pauli operators on the support of each lightcone $\mathsf{L}_i$,\footnote{This set of jump operators which ``drive'' the transition can be essentially arbitrary, however, this choice resembling the connectivity of the underlying Hamiltonian will play an important role in our analysis.}
\begin{equation}
    \{A^a\}_{a\in \mathcal{A}} = 2^{-\ell}\cdot  \bigg\{P_{\mathsf{L}_i}\otimes \mathbb{I}_{[n]\setminus \mathsf{L}_i}: i\in [n], P\in \mathcal{P}_{\ell}\bigg\},
\end{equation}
where $\mathcal{P}_{\ell}=\{I,X,Y,Z\}^{\otimes \ell}$.\footnote{Note that there are $|\mc A|=n\cdot 4^\ell$ jump operators.} In contrast to classical Markov Chain transitions, these quantum jumps will change the energy of the system in superposition. Thereby, it will be convenient to decompose the jump operators into the energy basis:
\begin{align}
A^a_\nu := \sum_{k\in [n]} \Pi_{k+\nu} A^a \Pi_k \quad \text{such that}\quad \sum_{\nu\in [-n, n]} A^a_\nu = A^a.  
\end{align}

\item \textbf{Transition Weights.} The transition weight is selected to be the Glauber dynamics weight, $\gamma(\nu) = 1/(1+e^{-\beta \nu})$ for all $\nu\in[-n,n]$.
\end{itemize}

\noindent Put together, the associated family of Davies generators $\mathcal{L}$ can be written down as\footnote{Where $\{A,B\}:=AB+BA$ is the anti-commutator.}
\begin{equation}\label{equation:Lindbladian}
    \mathcal{L}[\rho] = \sum_{a \in \mathcal{A}} \sum_{\nu}\gamma(\nu) \bigg(A^a_\nu \rho (A^a_\nu)^\dagger -\frac{1}{2} \bigg\{(A^a_\nu)^\dagger A^a_\nu, \rho \bigg\}\bigg).
\end{equation}

This construction satisfies the \textit{quantum detailed balance} condition, which implies that the desired Gibbs state is a fixed point $\mathcal{L}[\rho_\beta] = 0$ of the evolutions (see e.g. \cite{Wocjan2021SzegedyWU}, or \cref{fact:detailed-balance}). It remains to show that the Lindblad dynamics, governed by the exponential map

\begin{equation}
    \frac{d}{dt}\rho=\mathcal{L}[\rho] \Rightarrow \rho(t) = e^{\mathcal{L}t}[\rho_0],
\end{equation}

\noindent converges quickly to $\rho_\beta$. This is achieved by presenting a bound on the mixing time of $\mc L$, which is the shortest time $t_{mix}$ such that 
\begin{equation}
    \left\|e^{\mathcal{L}t_{mix}} [\rho-\sigma]\right\|_1\leq \frac{1}{2}\left\|\rho-\sigma\right\|_1,\quad \text{for all density matrices}\,\,\rho, \sigma.
\end{equation}

\subsubsection{A lightcone argument for the modified log-Sobolev inequality}
To study the mixing time of our algorithm, our starting point is first to study the trivial non-interacting Hamiltonian $H_\noti = \sum_{i\in[n]}\ketbra{1}_i\otimes \mathbb{I}_{[n]\setminus \{i\}}$, and prove a rapid mixing bound for the associated Davies generator $\mc L_\noti$. Subsequently, we argue that the mixing time of $\mc L$ can be \textit{compared} with that of $\mc L_\noti$. This is achieved by leveraging the lightcone structure of shallow quantum circuits. We begin by presenting basic definitions of Log-Sobolev bounds. 


\paragraph{Mixing time bounds via log-Sobolev inequalities.} There are two general purpose methods to bound the mixing time of Lindbladian evolution. The first of which consists of a bound on the spectral gap of $\mathcal{L}$. Unfortunately, a spectral gap bound comes at an inherent polynomial overhead to the mixing time, see \cref{section:gibbs_prep_appendix}. Instead, we make use of a much sharper notion of convergence known as a modified log-Sobolev inequality (MLSI) \cite{Kastoryano2012QuantumLS}. Informally, a MLSI quantifies the rate of decay of the relative entropy,\footnote{The quantum relative entropy between two density matrices $\rho, \sigma$ is given by $D(\rho||\sigma) = \Tr[\rho\cdot (\log \rho-\log \sigma\big)]$.} by relating it to the relative entropy itself:
\begin{equation}\label{eq:MLSI}
      \frac{d}{dt}\bigg|_{t=0}D\big(e^{t\mathcal{L}}[\rho] ||\rho_\beta\big) \leq -\alpha \cdot D\big(\rho||\rho_\beta\big)\quad \text{for every density matrix}\,\,\rho,  \tag*{(MLSI)}
\end{equation}
\noindent where $\alpha$ is known as the MLSI constant. This clearly implies an exponential decay of $D(e^{\mc L t}[\rho]||\rho_\beta)\leq e^{-\alpha t} \cdot D(\rho||\rho_\beta)$. Which, in turn, tells us the mixing time is bounded by $t_{mix}\leq \alpha^{-1}\cdot O(\log n)$ (Pinsker's inequality). This logarithmic mixing time bound is known as \textit{rapid mixing}, and proving good lower bounds on the constant $\alpha$ has proven to be quite challenging in the quantum setting. 

\paragraph{The non-interacting Lindbladian.} The simplest Hamiltonian in the family $\mathscr{H}$ is the non-interacting system $H_\noti$. Its Gibbs state is the tensor product state $\sigma_\beta\propto \big(e^{-\beta\ketbra{1}}\big)^{\otimes n}$. Under our framework (described in \cref{equation:Lindbladian}), its associated Lindbladian $\mc L_\noti$ has the same form as $\mc L$, except that the circuit $C$ has been replaced by the identity. In this manner, $\mc L_\noti$ itself can also be written as a sum of non-interacting, single-qubit components:
\begin{equation}
    \mathcal{L}_{\noti} = \sum_{i\in [n]} \mathcal{L}_{single}^i\otimes \mathbb{I}_{[n]\setminus \{i\}}
\end{equation}
Since each single qubit Lindbladian $\mathcal{L}_{single}$ is highly explicit (it acts on $2\times 2$ matrices), in \cref{section:gibbs_prep_appendix}, following now standard techniques, we are able to prove simple bounds on its MLSI constant. 


\begin{claim}
    $\mathcal{L}_{\noti}$ satisfies a MLSI with constant $\Omega(e^{-\beta})$.
\end{claim}

\paragraph{The convex combination argument.} The main technical challenge in our analysis lies in relating $\mathcal{L}_\noti$ with our family of Davies generators $\mathcal{L}$ from \cref{equation:Lindbladian}, in order to inherit the rapid mixing properties from the former. The crux of our proof lies in analyzing $\mc L$ in a basis rotated by $C$, to show that the rotated Davies generator is a \textit{convex combination} of $\mathcal{L}_\noti$ and some other Davies generator. This involves a delicate lightcone argument shown in \cref{fig:lightconeargument}, and discussed shortly.

\begin{claim}
    \label{claim:overview-convex-combination} In a basis rotated by $C$, the Lindbladian $\mathcal{L}$ from \cref{equation:Lindbladian} can be written as a convex combination
    \begin{equation}
    \Tilde{\mathcal{L}} \equiv   C^\dagger \mathcal{L}[C\cdot C^\dagger]C = q\cdot \mathcal{L}_{\noti}[\cdot] +(1-q) \cdot \mathcal{L}_{rest}[\cdot ],
    \end{equation}
    where both $\mathcal{L}_{\noti}, \mathcal{L}_{rest}$ share the fixed point $ \sigma_\beta$, and $q = 4^{1-\ell}$. 
\end{claim}

We emphasize that the parameter $q\in [0, 1]$ only depends on the lightcone size of $C$. Moreover, while $\mathcal{L}_{\noti}$ is the well-understood non-interacting system discussed previously, $\mathcal{L}_{rest}$ may apriori be arbitrary. However, at the very least we know it shares a fixed point with $\mathcal{L}_{\noti}$.

Neverthless, in \cref{section:gibbs_state_prep}  we show that convexity is precisely enough\footnote{This is inspired by \cite{Onorati2016MixingPO}, who leveraged the concavity of the spectral gap to prove mixing properties of stochastic Hamiltonians.} to exhibit an MLSI for $\mathcal{L}$, with a constant which is only a multiplicative factor of $q$ off of that of $\mathcal{L}_\noti$.

\begin{lemma}\label{lemma:rapid-mixing}
    $\mathcal{L}$ satisfies a MLSI with constant $\Omega(4^{-\ell}e^{-\beta})$.
\end{lemma}


\paragraph{The lightcone argument.} We conclude by discussing the key technical step: a proof of \cref{claim:overview-convex-combination} via a lightcone argument. The starting point is to examine the Davies generator $\mc L$ (\cref{equation:Lindbladian}) and its rotated version $\Tilde{\mathcal{L}}=C^\dagger \mathcal{L}[C\cdot C^\dagger]C$. The goal is to show that ``a fraction of'' $\Tilde{\mathcal{L}}$ equals the Davies generator of the non-interacting Hamiltonian $H_\noti$, that is, within that fraction, the effect of $C$ is erased.

To begin, note that the Davies generator $\mathcal{L}$ (\cref{equation:Lindbladian}) only depends on the circuit $C$ through the jump operators decomposed into the frequency basis, $A_\nu^a$, which we recollect can be written as

\begin{equation}
    A^a_\nu = \sum_k \Pi_{k+\nu} A^a \Pi_k=\sum_k C\bigg(\sum_{|y| = k+\nu}\ketbra{y} \bigg) C^\dagger A^a C\bigg(\sum_{|x| = k}\ketbra{x} \bigg) C^\dagger.
\end{equation}
Crucially, due to the rotation of $\mathcal{L}$ to $\Tilde{\mathcal{L}}$,\footnote{And the fact that our jump operators are Pauli operators.} we observe that the dependence of $C$ within $\Tilde{\mathcal{L}}$ is only through second-moment operators of the form
\begin{equation}
    \E_{P\sim\mc P_\ell} [C^\dag P C\otimes C^\dag P C],
\end{equation}
where we consider the sum of all jump operators that act on a specific lightcone $\mathsf{L}_i$ of size $\ell$, and recall that the jump operators are $\ell$-qubit Pauli operators. It remains to express this operator as a convex combination, as shown in \cref{fig:lightconeargument}. A sketch of the argument follows:
\begin{itemize}
    \item Step (i): Uses the identity $\E_{P\sim\mc P_\ell}\left[P\otimes P\right]=\frac{1}{2^\ell}\cdot  \mathrm{SWAP}$, and linearity of expectation. 
    \item Step (ii): Since $C$ is a low-depth circuit, one can cancel the quantum gates within the lightcone of qubit $i$ with their inverse.
    \item Step (iii): Uses the identity $\E_{P\sim\mc P_\ell}\left[P\otimes P\right]=\frac{1}{2^\ell}\cdot  \mathrm{SWAP}$ again, but in the other direction.
    \item Step (iv): Re-writes the expectation into two parts: the first part is over the 4 single-qubit Paulis that act only on qubit $i$, and the second part, is all the remaining $\ell$-qubit Paulis. 
\end{itemize}

The crux of the argument lies in noting that in the first part of Step (iv), the $i$th qubit has been completely disentangled from the remaining circuit. Thereby, the single-qubit Pauli acts on a disentangled wire, and all remaining gates cancel with each other.

This gives the desired convex combination, where the first term corresponds to the non-interacting Hamiltonian with single-qubit jump operators. Finally, note that our choice of the jump operators (as $\ell$-qubit Paulis acting on each lightcone) is crucial for this argument, and it is unclear if an arbitrary choice of jump operators would suffice.

\begin{figure}[t]
    \centering
\begin{equation*}
\begin{aligned}
&\E_{P\sim\mc P_\ell}
\begin{tikzpicture}[baseline={(0, -0.75*\scale*\tikzlength cm-\MathAxis pt)},x=\scale*\tikzlength cm,y=\scale*\tikzlength cm]
        \draw[black] (0,0) rectangle node{$C^\dag$} (5,1);
  \foreach \x in {1,2,...,5} {
        \draw[black] (\x-0.5, 0) -- (\x-0.5,-0.5);
        \draw[black] (\x-0.5, 1) -- (\x-0.5,1.5);
    }
    
\draw[black] (0,-1.5) rectangle node{$C$} (5,-2.5);
  \foreach \x in {1,2,...,5} {
        \draw[black] (\x-0.5, -0.5) -- (\x-0.5,-1.5);
        \draw[black] (\x-0.5, -2.5) -- (\x-0.5,-3);
    }
\fill[blue!30!white] (1.25,-0.5) rectangle (3.75,-1);
    \node at (2.5,-0.75) {\small $P$};
    \end{tikzpicture}
\otimes 
\begin{tikzpicture}[baseline={(0, -0.75*\scale*\tikzlength cm-\MathAxis pt)},x=\scale*\tikzlength cm,y=\scale*\tikzlength cm]
        \draw[black] (0,0) rectangle node{$C^\dag$} (5,1);
  \foreach \x in {1,2,...,5} {
        \draw[black] (\x-0.5, 0) -- (\x-0.5,-0.5);
        \draw[black] (\x-0.5, 1) -- (\x-0.5,1.5);
    }
    
\draw[black] (0,-1.5) rectangle node{$C$} (5,-2.5);
  \foreach \x in {1,2,...,5} {
        \draw[black] (\x-0.5, -0.5) -- (\x-0.5,-1.5);
        \draw[black] (\x-0.5, -2.5) -- (\x-0.5,-3);
    }
\fill[blue!30!white] (1.25,-0.5) rectangle (3.75,-1);
    \node at (2.5,-0.75) {\small $P$};
    \end{tikzpicture}
\overset{(i)}{=}\frac{1}{2^\ell}\cdot
\begin{tikzpicture}[baseline={(0, -0.75*\scale*\tikzlength cm-\MathAxis pt)},x=\scale*\tikzlength cm,y=\scale*\tikzlength cm]
        \draw[black] (0,0) rectangle node{$C^\dag$} (5,1);
  \foreach \x in {1,2,...,5} {
        \draw[black] (\x-0.5, 1) -- (\x-0.5,1.5);
    }
    
\draw[black] (0,-1.5) rectangle node{$C$} (5,-2.5);
  \foreach \x in {1,2,...,5} {
        \draw[black] (\x-0.5, -2.5) -- (\x-0.5,-3);
    }

        \draw[black] (5.5,0) rectangle node{$C^\dag$} (10.5,1);
  \foreach \x in {1,2,...,5} {
        \draw[black] (5.5+\x-0.5, 1) -- (5.5+\x-0.5,1.5);
    }
    
\draw[black] (5.5,-1.5) rectangle node{$C$} (10.5,-2.5);
  \foreach \x in {1,2,...,5} {
        \draw[black] (5.5+\x-0.5, -2.5) -- (5.5+\x-0.5,-3);
    }
\draw[black] (1-0.5, 0) -- (1-0.5,-1.5);
\draw[black] (5.5+1-0.5, 0) -- (5.5+1-0.5,-1.5);
\draw[black] (5-0.5, 0) -- (5-0.5,-1.5);
\draw[black] (9.5+1-0.5, 0) -- (9.5+1-0.5,-1.5);
\foreach \x in {1,2,...,3} {
    \draw[black] (\x+0.5,0)..controls (\x+0.5,-0.5) and (\x+6,-1) ..(\x+6,-1.5);
    \draw[black] (\x+0.5,-1.5)..controls (\x+0.5,-1) and (\x+6,-0.5) ..(\x+6,0);
}
    \end{tikzpicture}
\\[10pt]
&\overset{(ii)}{=}\frac{1}{2^\ell}\cdot\begin{tikzpicture}[baseline={(0, -0.75*\scale*\tikzlength cm-\MathAxis pt)},x=\scale*\tikzlength cm,y=\scale*\tikzlength cm]

\draw[black] (0,0) -- (1,0) -- (2,1) -- (0,1) -- (0,0);
\draw[black] (5,0) -- (4,0) -- (3,1) -- (5,1) -- (5,0);
\draw[black] (1.5,0) -- (1.5,0.5);
\draw[black] (2.5,0) -- (2.5,1);
\draw[black] (3.5,0) -- (3.5,0.5);

\draw[black] (5.5,0) -- (6.5,0) -- (7.5,1) -- (5.5,1) -- (5.5,0);
\draw[black] (10.5,0) -- (9.5,0) -- (8.5,1) -- (10.5,1) -- (10.5,0);
\draw[black] (5.5+1.5,0) -- (5.5+1.5,0.5);
\draw[black] (5.5+2.5,0) -- (5.5+2.5,1);
\draw[black] (5.5+3.5,0) -- (5.5+3.5,0.5);

\draw[black] (0,-2.5) -- (2,-2.5) -- (1,-1.5) -- (0,-1.5) -- (0,-2.5);
\draw[black] (5,-2.5) -- (3,-2.5) -- (4,-1.5) -- (5,-1.5) -- (5,-2.5);
\draw[black] (1.5,-1.5) -- (1.5,-2);
\draw[black] (2.5,-1.5) -- (2.5,-2.5);
\draw[black] (3.5,-1.5) -- (3.5,-2);

\draw[black] (5.5+0,-2.5) -- (5.5+2,-2.5) -- (5.5+1,-1.5) -- (5.5+0,-1.5) -- (5.5+0,-2.5);
\draw[black] (5.5+5,-2.5) -- (5.5+3,-2.5) -- (5.5+4,-1.5) -- (5.5+5,-1.5) -- (5.5+5,-2.5);
\draw[black] (5.5+1.5,-1.5) -- (5.5+1.5,-2);
\draw[black] (5.5+2.5,-1.5) -- (5.5+2.5,-2.5);
\draw[black] (5.5+3.5,-1.5) -- (5.5+3.5,-2);

  \foreach \x in {1,2,...,5} {
        \draw[black] (\x-0.5, 1) -- (\x-0.5,1.5);
    }
  \foreach \x in {1,2,...,5} {
        \draw[black] (\x-0.5, -2.5) -- (\x-0.5,-3);
    }

  \foreach \x in {1,2,...,5} {

        \draw[black] (5.5+\x-0.5, 1) -- (5.5+\x-0.5,1.5);
    }
    
  \foreach \x in {1,2,...,5} {
        \draw[black] (5.5+\x-0.5, -2.5) -- (5.5+\x-0.5,-3);
    }
\draw[black] (1-0.5, 0) -- (1-0.5,-1.5);
\draw[black] (5.5+1-0.5, 0) -- (5.5+1-0.5,-1.5);
\draw[black] (5-0.5, 0) -- (5-0.5,-1.5);
\draw[black] (9.5+1-0.5, 0) -- (9.5+1-0.5,-1.5);
\foreach \x in {1,2,...,3} {
    \draw[black] (\x+0.5,0)..controls (\x+0.5,-0.5) and (\x+6,-1) ..(\x+6,-1.5);
    \draw[black] (\x+0.5,-1.5)..controls (\x+0.5,-1) and (\x+6,-0.5) ..(\x+6,0);
}
    \end{tikzpicture}
\overset{(iii)}{=}\E_{P\sim\mc P_\ell}\begin{tikzpicture}[baseline={(0, -0.75*\scale*\tikzlength cm-\MathAxis pt)},x=\scale*\tikzlength cm,y=\scale*\tikzlength cm]

\draw[black] (0,0) -- (1,0) -- (2,1) -- (0,1) -- (0,0);
\draw[black] (5,0) -- (4,0) -- (3,1) -- (5,1) -- (5,0);
\draw[black] (1.5,0) -- (1.5,0.5);
\draw[black] (2.5,0) -- (2.5,1);
\draw[black] (3.5,0) -- (3.5,0.5);

\draw[black] (5.5,0) -- (6.5,0) -- (7.5,1) -- (5.5,1) -- (5.5,0);
\draw[black] (10.5,0) -- (9.5,0) -- (8.5,1) -- (10.5,1) -- (10.5,0);
\draw[black] (5.5+1.5,0) -- (5.5+1.5,0.5);
\draw[black] (5.5+2.5,0) -- (5.5+2.5,1);
\draw[black] (5.5+3.5,0) -- (5.5+3.5,0.5);

\draw[black] (0,-2.5) -- (2,-2.5) -- (1,-1.5) -- (0,-1.5) -- (0,-2.5);
\draw[black] (5,-2.5) -- (3,-2.5) -- (4,-1.5) -- (5,-1.5) -- (5,-2.5);
\draw[black] (1.5,-1.5) -- (1.5,-2);
\draw[black] (2.5,-1.5) -- (2.5,-2.5);
\draw[black] (3.5,-1.5) -- (3.5,-2);

\draw[black] (5.5+0,-2.5) -- (5.5+2,-2.5) -- (5.5+1,-1.5) -- (5.5+0,-1.5) -- (5.5+0,-2.5);
\draw[black] (5.5+5,-2.5) -- (5.5+3,-2.5) -- (5.5+4,-1.5) -- (5.5+5,-1.5) -- (5.5+5,-2.5);
\draw[black] (5.5+1.5,-1.5) -- (5.5+1.5,-2);
\draw[black] (5.5+2.5,-1.5) -- (5.5+2.5,-2.5);
\draw[black] (5.5+3.5,-1.5) -- (5.5+3.5,-2);

  \foreach \x in {1,2,...,5} {
        \draw[black] (\x-0.5, 1) -- (\x-0.5,1.5);
    }
  \foreach \x in {1,2,...,5} {
        \draw[black] (\x-0.5, -2.5) -- (\x-0.5,-3);
    }

  \foreach \x in {1,2,...,5} {

        \draw[black] (5.5+\x-0.5, 1) -- (5.5+\x-0.5,1.5);
    }
    
  \foreach \x in {1,2,...,5} {
        \draw[black] (5.5+\x-0.5, -2.5) -- (5.5+\x-0.5,-3);
    }
\draw[black] (1-0.5, 0) -- (1-0.5,-1.5);
\draw[black] (5.5+1-0.5, 0) -- (5.5+1-0.5,-1.5);
\draw[black] (5-0.5, 0) -- (5-0.5,-1.5);
\draw[black] (9.5+1-0.5, 0) -- (9.5+1-0.5,-1.5);
\fill[blue!30!white] (1.25,-0.5) rectangle (3.75,-1);
\node at (2.5,-0.75) {\small $P$};
\fill[blue!30!white] (5.5+1.25,-0.5) rectangle (5.5+3.75,-1);
\node at (5.5+2.5,-0.75) {\small $P$};
\foreach \x in {1,2,...,3} {
    \draw[black] (\x+0.5,0) -- (\x+0.5,-0.5);
    \draw[black] (5.5+\x+0.5,0) -- (5.5+\x+0.5,-0.5);
    \draw[black] (\x+0.5,-1) -- (\x+0.5,-1.5);
    \draw[black] (5.5+\x+0.5,-1) -- (5.5+\x+0.5,-1.5);
}
\end{tikzpicture}
\\[10pt]
&\overset{(iv)}{=}\frac{1}{4^{\ell-1}}\cdot\E_{P\sim\mc P_i}\,\begin{tikzpicture}[baseline={(0, -0.75*\scale*\tikzlength cm-\MathAxis pt)},x=0.8*\scale*\tikzlength cm,y=\scale*\tikzlength cm]
\foreach \x in {1,2,...,5} {
        \draw[black] (\x-0.5, 1.5) -- (\x-0.5,-3);
        \draw[black] (5.5+\x-0.5, 1.5) -- (5.5+\x-0.5,-3);
    }

\fill[blue!30!white] (2.1875,-0.5) rectangle (2.8125,-1);
\node at (2.5,-0.75) {\small $P$};
\fill[blue!30!white] (5.5+2.1875,-0.5) rectangle (5.5+2.8125,-1);
\node at (5.5+2.5,-0.75) {\small $P$};

\end{tikzpicture}
\,+
\frac{4^{\ell-1}-1}{4^{\ell-1}}\cdot\E_{Q\sim\mc P_\ell\setminus \mc P_i}\begin{tikzpicture}[baseline={(0, -0.75*\scale*\tikzlength cm-\MathAxis pt)},x=\scale*\tikzlength cm,y=\scale*\tikzlength cm]

\draw[black] (0,0) -- (1,0) -- (2,1) -- (0,1) -- (0,0);
\draw[black] (5,0) -- (4,0) -- (3,1) -- (5,1) -- (5,0);
\draw[black] (1.5,0) -- (1.5,0.5);
\draw[black] (2.5,0) -- (2.5,1);
\draw[black] (3.5,0) -- (3.5,0.5);

\draw[black] (5.5,0) -- (6.5,0) -- (7.5,1) -- (5.5,1) -- (5.5,0);
\draw[black] (10.5,0) -- (9.5,0) -- (8.5,1) -- (10.5,1) -- (10.5,0);
\draw[black] (5.5+1.5,0) -- (5.5+1.5,0.5);
\draw[black] (5.5+2.5,0) -- (5.5+2.5,1);
\draw[black] (5.5+3.5,0) -- (5.5+3.5,0.5);

\draw[black] (0,-2.5) -- (2,-2.5) -- (1,-1.5) -- (0,-1.5) -- (0,-2.5);
\draw[black] (5,-2.5) -- (3,-2.5) -- (4,-1.5) -- (5,-1.5) -- (5,-2.5);
\draw[black] (1.5,-1.5) -- (1.5,-2);
\draw[black] (2.5,-1.5) -- (2.5,-2.5);
\draw[black] (3.5,-1.5) -- (3.5,-2);

\draw[black] (5.5+0,-2.5) -- (5.5+2,-2.5) -- (5.5+1,-1.5) -- (5.5+0,-1.5) -- (5.5+0,-2.5);
\draw[black] (5.5+5,-2.5) -- (5.5+3,-2.5) -- (5.5+4,-1.5) -- (5.5+5,-1.5) -- (5.5+5,-2.5);
\draw[black] (5.5+1.5,-1.5) -- (5.5+1.5,-2);
\draw[black] (5.5+2.5,-1.5) -- (5.5+2.5,-2.5);
\draw[black] (5.5+3.5,-1.5) -- (5.5+3.5,-2);

  \foreach \x in {1,2,...,5} {
        \draw[black] (\x-0.5, 1) -- (\x-0.5,1.5);
    }
  \foreach \x in {1,2,...,5} {
        \draw[black] (\x-0.5, -2.5) -- (\x-0.5,-3);
    }

  \foreach \x in {1,2,...,5} {

        \draw[black] (5.5+\x-0.5, 1) -- (5.5+\x-0.5,1.5);
    }
    
  \foreach \x in {1,2,...,5} {
        \draw[black] (5.5+\x-0.5, -2.5) -- (5.5+\x-0.5,-3);
    }
\draw[black] (1-0.5, 0) -- (1-0.5,-1.5);
\draw[black] (5.5+1-0.5, 0) -- (5.5+1-0.5,-1.5);
\draw[black] (5-0.5, 0) -- (5-0.5,-1.5);
\draw[black] (9.5+1-0.5, 0) -- (9.5+1-0.5,-1.5);
\fill[green!30!white] (1.25,-0.5) rectangle (3.75,-1);
\node at (2.5,-0.75) {\small $Q$};
\fill[green!30!white] (5.5+1.25,-0.5) rectangle (5.5+3.75,-1);
\node at (5.5+2.5,-0.75) {\small $Q$};
\foreach \x in {1,2,...,3} {
    \draw[black] (\x+0.5,0) -- (\x+0.5,-0.5);
    \draw[black] (5.5+\x+0.5,0) -- (5.5+\x+0.5,-0.5);
    \draw[black] (\x+0.5,-1) -- (\x+0.5,-1.5);
    \draw[black] (5.5+\x+0.5,-1) -- (5.5+\x+0.5,-1.5);
}
\end{tikzpicture}
\end{aligned}
\end{equation*}
    \caption{A lightcone argument for proving the modified log-Sobolev inequality.}
    \label{fig:lightconeargument}
\end{figure}

\subsubsection{Efficient implementation on a quantum computer}

Rapid mixing of the Davies generator implies that the Gibbs state can be efficiently prepared in the thermal model of computation, described by coupling the quantum system to a thermal bath \cite{Kliesch_2011}. Next, we briefly discuss how to simulate the dissipative Lindbladian evolution $e^{\mathcal{L}t}$ on a quantum computer.

We leverage the ``continuous-time quantum Gibbs sampler" framework of \cite{Chen2023QuantumTS}. They show that to implement the map $e^{\mathcal{L}t}$ one requires $\Tilde{O}(t)$ black-box invocations to a \textit{unitary block-encoding} of the Lindblad operators (\cite{Chen2023QuantumTS}, Theorem I.1). In turn, to implement such a block-encoding for Hamiltonians $H$ of integer spectra, it suffices to design quantum circuits which implement the Hamiltonian simulation of $H$, a block-encoding for the jump operators $A^a$, as well as a certain``frequency filter" which implements the Glauber dynamics weight. In \cref{section:implementation}, we discuss circuit implementations of all these ingredients, summarized in the following Lemma.




\begin{lemma}[Dissipative Lindbladian Implementation]\label{lemma:dissipative_implementation_overview}
    Fix parameters $t \geq  1$ and $\epsilon\leq \frac{1}{2}$. Let $\mathcal{L}$ denote the Lindbladian of \cref{equation:Lindbladian}, defined by a quantum circuit $C$ on $n$ qubits of depth $d$ and lightcone size $\ell$. Then, we can simulate the map $e^{t\mathcal{L}}$ to error $\epsilon$ in diamond norm using a quantum circuit of depth $O(t\cdot n\cdot 4^{\ell}\cdot 2^d \cdot \poly(\ell, \log n, \log \frac{1}{\epsilon}, \log t))$.
\end{lemma}

\noindent Put together with our bound on the mixing time, we arrive at our main statement on Gibbs state preparation in \cref{lemma:results-gibbs-prep}.

\subsection{Classical hardness of Gibbs sampling}
\label{section:hardness}

As discussed in \cref{section:introduction}, to obtain the classical intractability of quantum Gibbs sampling it suffices to construct a family of low depth quantum circuits which are hard to sample from even in the presence of input errors (\cref{fig:noisemodel}). The reason this imposes a challenge is two-fold. First, it is known that many classically hard shallow quantum circuits actually become classically simulable under input noise~\cite{Bremner2016AchievingQS}, thereby suggesting a need for fault-tolerance techniques. However, standard fault-tolerance techniques \cite{Aharonov-Ben-Or} often come with a prohibitive circuit depth overhead, which blows up the locality of the parent Hamiltonian. We address these challenges by designing a fault-tolerance scheme tailored to the input noise model with small overhead.

Our plan is to focus on IQP circuits, which are known to be already classically hard at constant depth. We show that their commuting structure plays an important role in our fault-tolerance techniques at low overhead.

\subsubsection{Quantum computational advantage with shallow IQP circuits}
\label{section:overview-advantage}

Recall that IQP circuits can be written as $C = H^{\otimes n} D H^{\otimes n}$, where $D$ is a diagonal unitary. The induced probability distribution $p(x)=|\mel{x}{C}{0^n}|^2$ is hard to sample from classically in general \cite{Bremner2010ClassicalSO}. While any family of constant-depth and classically-hard IQP circuits suffices for our purpose, in this paper we use the concrete example of cluster states on regular lattices composed with random $Z$-rotations\footnote{In the literature, these circuits are also known as the ``evolution (quench) of an nearest-neighbor, translationally invariant (NNTI) Hamiltonian''.}, which have become the basis for various proposals of sampling-based quantum supremacy using low-depth circuits \cite{Bremner2010ClassicalSO, Bremner2016AchievingQS, Gao2016QuantumSF, BermejoVega2017ArchitecturesFQ, Hangleiter2017AnticoncentrationTF, Novo2019QuantumAF}.

We present the structure of these circuits in more detail in \cref{section:appendix-hardness}, where we additionally present a comprehensive discussion on the foundations of their hardness. As a brief overview, note that 2D cluster states with single-qubit $Z$ rotations is a universal resource state for measurement-based quantum computation (MBQC)~\cite{Broadbent2008UniversalBQ}. This implies that \textit{exactly} sampling from their output distribution is hard in the worst-case \cite{Aaronson2010TheCC}. The hardness of approximate sampling from these architectures are based on further assumptions \cite{Bremner2015AveragecaseCV, Gao2016QuantumSF, Hangleiter2017AnticoncentrationTF}, which we rigorously define in \cref{section:appendix-hardness}. The following theorem thus provides the complexity-theoretic basis of our hardness arguments.

\begin{theorem}[Complexity of constant-depth IQP sampling~\cite{Gao2016QuantumSF, BermejoVega2017ArchitecturesFQ}]\label{theorem:sampling}
    There exists a constant $\delta > 0$, and a family of constant depth IQP circuits $\{C_n\}_{n\geq 1}$ on $n$ qubits, such that no randomized classical polynomial-time algorithm can sample from the output distribution of $C_n$ up to additive error $\delta$ in total variation distance, assuming the average-case hardness of computing a fixed family of partition functions (\cref{conjecture:mixture_of_errors}), and the non-collapse of the Polynomial Hierarchy.
\end{theorem}

To establish classical hardness of the Gibbs sampling task, it suffices to map the above circuit $C$ to a fault-tolerant circuit $\Tilde{C}$, such that a sample from the output distribution of $\Tilde{C}$ under input noise can be efficiently post-processed into an ideal sample from $C$. The key challenge is to reduce the fault-tolerance overhead in $\Tilde{C}$, so that the corresponding parent Hamiltonian has small locality.

\subsubsection{Fault-tolerance of IQP circuits against input noise}
\label{section:overview-distillation}


The starting point in our approach is the observation that it suffices to \textit{error-detect} the random inputs bits, instead of correcting them, to preserve the hardness-of-sampling of $C$. Indeed, bit-flip errors (which are Pauli-$X$ errors) on the input of IQP circuits, become phase-flip errors after the first layer of Hadamard gates, and thus commute with the entire IQP circuit. In this manner, they are equivalent to bit-flip errors on the measured output string. Therefore, if we could identify the computational basis state $\ket{r} = \otimes_i \ket{r_i}$ fed into the IQP circuit, we would be able to correct the measured output sample by simply subtracting $r \in \{0, 1\}^n$. Indeed, we emphasize we don't intend to correct the input error within the quantum circuit at all, as this would require decoding and feedforward, and potentially a much deeper circuit. Instead, we correct the error only during classical post-processing (that is, after all qubits are measured). 

The crux of our approach is the design of a ``distillation" gadget, which independently pre-processes each input bit $r_i$ into $k$ others in such a manner which enables us to reconstruct $r_i$ (with high probability) given only the other $k-1$ noisy bits. We illustrate this task with a simple example, based on the repetition code.

\begin{figure}[t]
    \centering
    \begin{subfigure}[b]{0.2\textwidth}
    \begin{tikzpicture}[baseline={(0, 0.5*\tikzlength cm-\MathAxis pt)},x=\tikzlength cm,y=\tikzlength cm]
    \draw  node[fill,circle,inner sep=2pt,minimum size=1pt] at (0,0) (1) {};
    \draw  node[fill,circle,inner sep=2pt,minimum size=1pt] at (-2,-3) (2) {};
    \draw  node[fill,circle,inner sep=2pt,minimum size=1pt] at (-1,-3) (3) {};
    \draw  node[fill,circle,inner sep=2pt,minimum size=1pt] at (0,-3) (4) {};
    \node at (1,-3) {$\cdots$};
    \draw  node[fill,circle,inner sep=2pt,minimum size=1pt] at (2,-3) (5) {};
    \draw[->,style=thick] (1) -- (2);
    \draw[->,style=thick] (1) -- (3);
    \draw[->,style=thick] (1) -- (4);
    \draw[->,style=thick] (1) -- (5);
    \end{tikzpicture}
    \caption{Repetition gadget}
    \label{fig:repetition}
    \end{subfigure}
    \begin{subfigure}[b]{0.4\textwidth}
    \begin{tikzpicture}[baseline={(0, 0.5*\tikzlength cm-\MathAxis pt)},x=\tikzlength cm,y=\tikzlength cm]
    \draw  node[fill=violet,circle,inner sep=2pt,minimum size=1pt] at (0,0) (1) {};
    \draw  node[fill=violet,circle,inner sep=2pt,minimum size=1pt] at (-2,-1.5) (2) {};
    \draw  node[fill=violet,circle,inner sep=2pt,minimum size=1pt] at (0,-1.5) (3) {};
    \draw  node[fill=violet,circle,inner sep=2pt,minimum size=1pt] at (2,-1.5) (5) {};
    \draw[->,style=thick] (1) -- (2);
    \draw[->,style=thick] (1) -- (3);
    \draw[->,style=thick] (1) -- (5);

    \draw node[fill,circle,inner sep=2pt,minimum size=1pt] at (-2.5,-3) (6) {};
    \draw node[fill,circle,inner sep=2pt,minimum size=1pt] at (-2,-3) (7) {};
    \draw node[fill,circle,inner sep=2pt,minimum size=1pt] at (-1.5,-3) (8) {};
    \draw[->,style=thick] (2) -- (6);
    \draw[->,style=thick] (2) -- (7);
    \draw[->,style=thick] (2) -- (8);

    \draw node[fill,circle,inner sep=2pt,minimum size=1pt] at (-0.5,-3) (9) {};
    \draw node[fill,circle,inner sep=2pt,minimum size=1pt] at (0,-3) (10) {};
    \draw node[fill,circle,inner sep=2pt,minimum size=1pt] at (0.5,-3) (11) {};
    \draw[->,style=thick] (3) -- (9);
    \draw[->,style=thick] (3) -- (10);
    \draw[->,style=thick] (3) -- (11);

    \draw node[fill=violet,circle,inner sep=2pt,minimum size=1pt] at (1.5,-3) (12) {};
    \draw node[fill=violet,circle,inner sep=2pt,minimum size=1pt] at (2,-3) (13) {};
    \draw node[fill=violet,circle,inner sep=2pt,minimum size=1pt] at (2.5,-3) (14) {};
    \draw[->,style=thick] (5) -- (12);
    \draw[->,style=thick] (5) -- (13);
    \draw[->,style=thick] (5) -- (14);

    \draw node[fill,circle,inner sep=2pt,minimum size=1pt] at (0,-4.5) (15) {};
    \draw node[fill,circle,inner sep=2pt,minimum size=1pt] at (0.5,-4.5) (16) {};
    \draw node[fill,circle,inner sep=2pt,minimum size=1pt] at (1,-4.5) (17) {};
    \draw[->,style=thick] (12) -- (15);
    \draw[->,style=thick] (12) -- (16);
    \draw[->,style=thick] (12) -- (17);

    \draw node[fill,circle,inner sep=2pt,minimum size=1pt] at (1.5,-4.5) (21) {};
    \draw node[fill,circle,inner sep=2pt,minimum size=1pt] at (2,-4.5) (22) {};
    \draw node[fill,circle,inner sep=2pt,minimum size=1pt] at (2.5,-4.5) (23) {};
    \draw[->,style=thick] (13) -- (21);
    \draw[->,style=thick] (13) -- (22);
    \draw[->,style=thick] (13) -- (23);

    \draw node[fill=violet,circle,inner sep=2pt,minimum size=1pt] at (3,-4.5) (18) {};
    \draw node[fill=violet,circle,inner sep=2pt,minimum size=1pt] at (3.5,-4.5) (19) {};
    \draw node[fill=violet,circle,inner sep=2pt,minimum size=1pt] at (4,-4.5) (20) {};
    \draw[->,style=thick] (14) -- (18);
    \draw[->,style=thick] (14) -- (19);
    \draw[->,style=thick] (14) -- (20);

    \node at (4,-4.9) {$\vdots$};
    \draw node[fill=violet,circle,inner sep=2pt,minimum size=1pt] at (4,-5.5) (24) {};
    \draw node[fill=violet,circle,inner sep=2pt,minimum size=1pt] at (4,-7) (25) {};
    \draw node[fill=violet,circle,inner sep=2pt,minimum size=1pt] at (4.5,-7) (26) {};
    \draw node[fill=orange,circle,inner sep=2pt,minimum size=1pt] at (5,-7) (27) {};
    \draw[->,style=thick] (24) -- (25);
    \draw[->,style=thick] (24) -- (26);
    \draw[->,style=thick] (24) -- (27);

    \node at (2,-4.9) {$\vdots$};
    \node at (0.5,-4.9) {$\vdots$};
    \node at (0,-3.5) {$\vdots$};
    \node at (-2,-3.5) {$\vdots$};
    \end{tikzpicture}
    \caption{Recursive concatenation}
    \label{fig:recursive}
    \end{subfigure}
    \begin{subfigure}[b]{0.3\textwidth}
    \begin{tikzpicture}[baseline={(0, 0.5*\tikzlength cm-\MathAxis pt)},x=\tikzlength cm,y=\tikzlength cm]
    \fill[blue!30!white,rounded corners=5pt]  (0.2,0.2) -- (0,-1.5) -- (-1.5,0) -- cycle;
    \fill[blue!30!white,rounded corners=5pt]  (3.8,0.2) -- (4,-1.5) -- (5.5,0) -- cycle;
    \fill[blue!30!white,rounded corners=5pt]  (2,0.3) -- (1,-1.5) -- (3,-1.5) -- cycle;
    \fill[blue!30!white,rounded corners=5pt]  (2,3.7) -- (1,5.5) -- (3,5.5) -- cycle;
    \fill[blue!30!white,rounded corners=5pt]  (3.7,2) -- (5.5,1) -- (5.5,3) -- cycle;
    \fill[blue!30!white,rounded corners=5pt]  (0.3,2) -- (-1.5,1) -- (-1.5,3) -- cycle;
    \fill[blue!30!white,rounded corners=5pt]  (3.8,3.8) -- (5.5,4) -- (4,5.5) -- cycle;
    \fill[blue!30!white,rounded corners=5pt]  (0.2,3.8) -- (-1.5,4) -- (0,5.5) -- cycle;
    \fill[blue!30!white,rounded corners=5pt]  (1.2,1.2) -- (2.8,1.2) -- (2.8,2.8) -- (1.2,2.8) -- cycle;
    
    \draw[step=2,black,very thick] (0,0) grid (4,4);
    \foreach \x in {0,...,2}
    \foreach \y in {0,...,2}
    {
        \filldraw[black] (2*\x,2*\y) circle (2pt);}
    \end{tikzpicture}
    \caption{Fault-tolerant circuit}
    \label{fig:ftcircuit}
    \end{subfigure}
    \caption{Fault-tolerance via state distillation gadgets. (a) The repetition code gadget. (b) A $B$-Tree and the recursive concatenation scheme. Arrows denote the direction of $\cnot$ gates. (c) Pre-processing the circuit using distillation gadgets.}
    \label{fig:distillation}
\end{figure}

\paragraph{A distillation gadget based on the repetition code.} Recall that all input bits are noisy: each of them is flipped from 0 to 1 with probability $q$. Given $k$ bits drawn from $s\leftarrow \ber^k(q)$, suppose we designate the $k$-th bit as the ``root" and apply a $\cnot$ gate from it to the other $k-1$ bits (\cref{fig:repetition}). During the decoding stage, we would like to reconstruct the root bit given the other $k-1$ bits. To do this we simply compute the majority of the ``leaves":
\begin{equation}
\begin{aligned}
    \text{Gadget}(s_1, s_2, \cdots, s_k) &= (s_1\oplus s_k, s_2\oplus s_k, \cdots, s_{k-1}\oplus s_k, s_k),\\
    \hat s_k&=\maj(s_1\oplus s_k, s_2\oplus s_k, \cdots, s_{k-1}\oplus s_k).
\end{aligned}
\end{equation}
We show that the probability of failure (when $\hat s_k\neq s_k$) equals $\delta = q^{\Omega(k)}$.

To highlight how these gadgets can be used for fault-tolerance, given an $n$-qubit IQP circuit $C$, we begin by pre-processing each of $n$ input bits independently into a distillation gadget of size $k$, resulting in a circuit on $n\cdot k$ bits. Each of the $n$ ``root" bits are then fed into $C$ (\cref{fig:ftcircuit}). Note that the $n\cdot (k-1)$ remaining bits are untouched by $C$. In the end, after all qubits are measured, we can use the $n\cdot (k-1)$ ancilla bits to infer if an error had happened on each of the ``root" bits fed into the circuit. As argued earlier, if an error did happen, it can be corrected by simply flipping the measurement outcome since the error commutes with the circuit. If we choose $k = \Theta(\log n)$, then the entire error correction process succeeds with high probability.

\paragraph{Recursive concatenation and $B$-Trees.} In effect, the scheme above distills the ``root" bit $s_k$ with an effective bit-flip error rate $q^{\Omega(k)}$, using $k-1$ redundant ``syndrome" bits of error rate $q$. Note that it used no information about the distribution of $s_k$, only that of the ``leaves" $s_1, \cdots, s_{k-1}$.

To improve on this example, we bootstrap the above technique by recursively preparing ``syndrome" bits of better and better fidelity.\footnote{This construction is largely inspired by recursive magic state distillation schemes.} Suppose we organize $k$ bits into a tree of arity $B$ and depth 2, such $k = 1+B+B^2$. Moreover, apply the repetition code gadget on each layer, from leaves to root of the tree, by applying a $\cnot$ gate from each parent bit to their respective children bits in the tree. In doing so, by the previous analysis we can identify each bit at the middle layer, just using the bits at the leaves, with error probability $q^{\Omega(B)}$. By performing majority again at the middle layer, we are now able to identify the bit at the root of this two-layer tree with error rate $(q^{\Omega(B)})^{\Omega(B)} = q^{\Omega(B^2)}$. By recursively applying this approach on a $B$-tree of depth $d$, the error probability at the root of the tree scales doubly-exponentially with the depth $d$, $q^{B^{\Omega(d)}}$. 

At face value, it may seem that we haven't gained anything over the repetition code, as the error probability still only decays exponentially with the size of the gadget. The advantage lies instead in the \textit{locality} of the gadget. Indeed, consider the lightcone of the orange qubit $u$ at the leaf of the tree in~\cref{fig:recursive}. By examining the causal influence of this qubit, we conclude that only the qubits in the neighborhood of its path to the root (the purple nodes in~\cref{fig:recursive}) can lie in its lightcone. That is, if 
\begin{equation}
    u=u_0\rightarrow  u_1 \rightarrow \cdots \rightarrow u_d\equiv \text{ root }
\end{equation}
denotes the path from leaf to root, then the \textit{lightcone} of $u$ is contained the union of the neighborhoods $\mathsf{L}_u\equiv \cup_i^d N(u_i)$. Therefore, $|\mathsf{L}_u|\leq O(B\cdot d)$, which is a linear function of the depth of the tree. By further studying the propagation of $Z$ Pauli's through the gadget, we analogously show that that the locality of the parent Hamiltonian of the distillation circuit is $|\mathsf{S}_u|\leq d$; precisely the nodes on the path from leaf to root. \cref{lemma:results-iqp-ft} then follows from a careful choice of $B$ and $d$.

\subsection{Organization}

We organize the rest of this work as follows. \\

\noindent \textbf{Gibbs State Preparation.} In \cref{section:gibbs_prep_appendix}, we prove our rapid mixing bounds for Davies Generators, and in \cref{section:implementation} discuss their simulation on a quantum computer. In \cref{sec:latticeprep}, we discuss alternative state preparation algorithms for parent Hamiltonians of circuits which lie on lattices. \\

\noindent \textbf{Classical Intractability of Gibbs Sampling.} In \cref{section:noise}, we prove that the constant temperature Gibbs states of the Hamiltonians in $\mathscr{H}$, can be interpreted as the output of noisy circuits. In \cref{section:appendix-hardness}, we present an overview of the computational complexity of shallow IQP sampling. In \cref{section:fault-tolerance}, we present our fault-tolerance scheme based on state distillation. \\

\noindent Finally, in \cref{section:together}, we put everything together and prove our main result (\cref{theorem:main}).\\

\noindent \textbf{Applications.} In \cref{sec:measurementnoise}, we present our results on Gibbs sampling with measurement errors, and in \cref{section:completeness}, we discuss the BQP completeness of Gibbs sampling with adaptive single-qubit measurements.

\section{Rapid Mixing and Efficient Gibbs State Preparation}
\label{section:gibbs_prep_appendix}

We dedicate this section to a proof of the rapid convergence of our dissipative Lindbladians. We defer a discussion on its implementation using quantum circuits to \cref{section:implementation}. For simplicity, henceforth we re-scale the class of parent Hamiltonians,\footnote{Note that we are simply applying an affine transformation, $H = \frac{1}{2}\big(n\cdot \mathbb{I} + C\sum Z_i C^\dagger\big)$, such that the Gibbs state of $H$ at temperature $\beta$ is the same as that of $C\sum Z_i C^\dagger$ at temperature $\beta/2$.}
\begin{equation}
    H=\sum_{i\in [n]} h_i = \sum_i C \big(\ketbra{1}\otimes \mathbb{I}_{[n]\setminus i}\big) C^\dagger 
\end{equation}

\noindent to ensure frustration-freeness and positive integer spectra $[n] = \{0, \cdots, n\}$. Recall this Hamiltonian is commuting, and its eigenstates are given by $\{C\ket{x}:x\in \{0, 1\}^n\}$. Let $\ell$ be the lightcone size of $C$. The jump operators of our Davies generator are $\ell$-qubit Pauli operators on the support of the circuit lightcone. We refer the reader to \cref{section:gibbs_state_prep}, \cref{equation:Lindbladian} for a description of our Lindbladian. 

\begin{remark}
    The support $\mathsf{S}_i$ of each Hamiltonian term $h_i$ is contained within the lightcone $\mathsf{L}_i$, see below \cref{eq:projeigenspace}. In general, $h_i$ acts nontrivially on the entire lightcone (for example, when $C$ uses Haar random gates), but for our hard instances $|\mathsf{S}_i|\ll |\mathsf{L}_i|$. To ensure the quantum algorithm always ``knows'' $\{\mathsf{L}_i\}$ when it only sees $\{h_i\}$, we assume this information is encoded in the definition of the family of Hamiltonians, and is given to the algorithm. 

\end{remark}

\subsection{Preliminaries on thermal Lindbladians and their convergence}

We dedicate this subsection to background on the evolution and convergence of open quantum systems described by a Lindbladian. Recall, a general Lindbladian is a continuous-time Markov chain acting on density operators
\begin{align}
    \mathcal{L}[\rho] = \sum_{j} J_j\rho J^{\dagger}_j- \frac{1}{2} \{J_jJ_j^{\dagger},\rho\} \quad \text{for some set of Lindblad operators}\quad \{J_j\}_j,
\end{align}
which generates a family of completely positive and trace-preserving map
\begin{align}
    e^{\mathcal{L}t}[\rho]\quad \text{for each}\quad t\ge 0.
\end{align}

Our Lindbladians of interest satisfy a particular property known as detailed balance.

\begin{definition}
    [$s$-Inner Product] Fix a full rank density matrix $\sigma$ and $s\in [0, 1]$. We define the weighted Hilbert-Schmidt inner product:
    \begin{equation}
        \langle A, B\rangle_s = \langle A, \sigma^{1-s}B\sigma^s\rangle = \Tr[A^\dagger \sigma^{1-s}B\sigma^s]\quad \text{for each}\quad A, B.
    \end{equation}
\end{definition}

\begin{definition}[$s$-Detailed Balance]
     A Lindbladian $\mathcal{L}$ is $s$-\emph{Detailed Balance} with respect to $\sigma$ if $\mathcal{L}^\dagger$ is self-adjoint with respect to $\langle\cdot, \cdot \rangle_s$:
     \begin{equation}
        \forall A, B: \langle A, \mathcal{L}^\dagger[B]\rangle_s =\langle \mathcal{L}^\dagger[A], B\rangle_s 
    \end{equation}
\end{definition}

There are two important structural consequences of this detailed balance condition. The first is that the density operator $\sigma$ is a fixed point of Lindbladian evolution:

\begin{equation}
    \mathcal{L}[\sigma] = 0
\end{equation}

The second, as discussed shortly, is that it implies a powerful means to understand the convergence of the mixing process. For the reader most familiar with classical Markov chains, the detailed balance condition above is an analog to its classical counterpart, however, with an additional degree of freedom $0\le s\le 1$ which arises due to non-commutativity.

Two special cases of the above are the GNS (where $s=1$) and KMS ($s=1/2$) detailed balance conditions. Fortunately, under minor constraints on the family of Lindbladians (which our Lindbladian satisfies), all these definitions collapse. We refer the reader back to \cref{equation:Lindbladian} for the definition of the family of Lindbladians we consider, Davies Generators. 

\begin{fact}[Davies' generators are detailed balanced] \label{fact:detailed-balance}
    Consider the Davies generator $\mathcal{L}$ described in \cref{equation:Lindbladian}, subject to the constraint that the transition weights satisfy $\forall\nu: \gamma(\nu)/\gamma(-\nu) = e^{-\beta \nu}$, and the jump operators contain their adjoints $\{A_a\} = \{A_a^\dagger\}$. Then, $\mathcal{L}$ satisfies $s$-DB $\forall s\in [0, 1]$ w.r.t. the Gibbs state $\rho_\beta\propto e^{-\beta  H}$. 
\end{fact}

In this manner, the Gibbs state $\rho_\beta\propto e^{-\beta  H}$ is a fixed point of the Davies generator we designed in \cref{section:gibbs_state_prep}. However, it may not be the unique stationary state, nor may its evolution converge rapidly. To understand the rate of convergence of this process, we need a bound on its mixing time $t_{mix}(\mathcal{L})$. Physically, the mixing time provides an estimate for the thermalization time of the system. 
\begin{definition}[Mixing time]
    The mixing time $t_{mix}(\mathcal{L})$ of a Lindbladian $\mathcal{L}$ is the smallest time $t \ge 0$ for which 
    \begin{equation}
        \|e^{t \mathcal{L}}(\rho_1-\rho_2)\|_1\leq \frac{1}{2}\|\rho_1-\rho_2\|_1\quad \text{ for any two states }\rho_1, \rho_2.
    \end{equation}
\end{definition}

\noindent In what remains of this subsection, we describe two means to analyze $t_{mix}$. The first of which consists of a bound on the spectral gap of $\mathcal{L}$. Apriori, however, the super-operator $\mathcal{L}$ is not even Hermitian, and its spectral gap may not even be well-defined.
Fortunately, under an appropriate similarity transformation, we can appeal to a related Hermitian quantity known as \textit{the discriminant}:

\begin{definition}[Quantum discriminant]\label{definition:discriminant}
    Fix $s\in [0, 1]$ and a full-rank density matrix $\sigma$. The discriminant $\mathcal{K}_s$ of $\mathcal{L}$ consists of the super-operator
    \begin{equation}
         \mathcal{K}_s(\cdot) = \sigma^{-\frac{1-s}{2}}\mathcal{L}\left(\sigma^{\frac{1-s}{2}}\cdot \sigma^{\frac{s}{2}}\right)\sigma^{-\frac{s}{2}}.
    \end{equation}
\end{definition}

\begin{lemma}[\cite{Wocjan2021SzegedyWU}, Lemma 5 and 7]\label{lemma:discriminant_properties}
    The discriminant $\mathcal{K}_s$ of $\mathcal{L}$ satisfies the following properties
    \begin{enumerate}
        \item $\mathcal{L}$ satisfies $s$-DB if and only if $\mathcal{K}_s$ is Hermitian.
        \item If $\mathcal{L}$ satisfies $s$-DB, then the eigenvalues of $\mathcal{L}$ are the same as that of $\mathcal{K}_s$, which are real. 
        \item If $\mathcal{L}$ is a Davies generator satisfying the constraints of \cref{fact:detailed-balance}, then $\mathcal{K}\equiv \mathcal{K}_s$ is independent of $s\in [0, 1]$. 
    \end{enumerate}
\end{lemma}

The spectral gap $\Delta(\mathcal{L}) = \Delta(\mathcal{K}_s)$ of a given Lindbladian is defined to be that of the associated discriminant. Analyzing this gap can be a challenging task, and concrete bounds are often case-dependent. Nevertheless, it provides a powerful means to control the convergence of the time-evolution.

 \begin{lemma}[Mixing time from the Spectral Gap, \cite{Kastoryano2012QuantumLS}]
    If a Lindbladian $\mathcal{L}$ satisfies KMS reversibility with fixed point $\sigma$, then 
    \begin{equation}
        t_{mix}(\mathcal{L}) \leq \frac{\log( 2\|\sigma^{-1/2}\|)}{\Delta(\mathcal{L})},
    \end{equation}
\end{lemma}

We remark that the dependence on $\log \|\sigma^{-1/2}\|\approx O(\beta n)$ often-times incurs a polynomial overhead to the mixing time. The notion of a (modified) Log Sobolev inequality provides a significantly stronger means of analyzing the mixing time. To formalize this method, we first require the definition of the conditional expectation of an operator $X$, $\mathcal{E}[X] = \lim_{t\rightarrow \infty} e^{t\mathcal{L}}[X].$

\begin{definition}[Modified Logarithmic Sobolev inequality]
    The Markov semigroup $(e^{t\mathcal{L}})_{t\geq 0}$ satisfies a \emph{Modified Logarithmic Sobolev inequality} (MSLI) with constant $\alpha$ if 
    \begin{equation}
      \frac{d}{dt}D\bigg(e^{t\mathcal{L}}[\rho]||\mathcal{E}[\rho]\bigg)\bigg|_{t=0}  = \Tr \mathcal{L}[\rho](\log \rho - \log \mathcal{E}[\rho])
       \leq -  \alpha\cdot D(\rho||\mathcal{E}[\rho]) \quad \text{for each}\quad \rho,
    \end{equation}
    \noindent where $D(\rho||\sigma) = \Tr \rho (\log \rho - \log \sigma)$ is the quantum relative entropy. 
\end{definition}

In other words, a MLSI quantifies the decay of the relative entropy, which converts to a bound on the mixing time through Pinsker's inequality.

\begin{lemma}
    [Mixing time from MLSI, \cite{Kastoryano2012QuantumLS}] If a Lindbladian $\mathcal{L}$ satisfies KMS-detailed balance with fixed point $\sigma$ and a MLSI with constant $\alpha$, then 
    \begin{equation}
        t_{mix}(\mathcal{L})\leq \frac{2\cdot \log (4\cdot \log \|\sigma^{-1}\|)}{\alpha}
    \end{equation}
\end{lemma}

This \textit{polylogarithmic} overhead in system size is known as \textit{rapid mixing}. Moreover, if given an additional entangled reference system $R$ the semigroup $(e^{t\mathcal{L}}\otimes \mathbb{I}_R)_{t\geq 0}$ satisfies an MSLI, then $\mathcal{L}$ is said to satisfy a \textit{complete} modified logarithmic Sobolev inequality (CMLSI).

\subsection{Analysis}

The main result of this subsection is a bound on the mixing time of our family of Lindbladians, 

\begin{lemma}\label{lemma:mixingtimebound}
    The mixing time of our family of Lindbladians $\mathcal{L}$ defined in \cref{equation:Lindbladian} is bounded by
    \begin{equation}
        t_{mix}(\mathcal{L}) = O(4^\ell \cdot  e^\beta\cdot \log n).
    \end{equation}
\end{lemma}

The starting point of our analysis is based on that of a much simpler Lindbladian, namely, that corresponding to the trivial circuit $C=\mathbb{I}$. In this setting, both the associated parent Hamiltonian, and the associated Lindbladian, are a sum over non-interacting, single-qubit terms:
\begin{equation}
   H_{\noti} = \sum_i \ketbra{1}_i \text{ and } \mathcal{L}_{\noti} = \sum_{i\in [n]} \mathcal{L}^{i}_{single}, \quad\text{where}\quad \mathcal{L}^{i}_{single}[\sigma_\beta^i] = 0\quad \text{and}\quad\sigma_\beta^i \propto e^{-\beta \ketbra{1}_i}.
\end{equation}

\noindent The jump operators of $\mathcal{L}_{single}^i$ are simply single-qubit Pauli operators, and the single-qubit Gibbs state $\sigma_\beta^i$ is its fixed point. Using now standard techniques, one can prove that this non-interacting Lindbladian is both gapped and mixes rapidly:

\begin{claim}
    [The Non-Interacting Lindbladian is rapidly mixing]\label{claim:analysis-ni-gapped} The non-interacting Lindbladian $\mathcal{L}_{\noti}$ has a constant spectral gap $\Delta(\mathcal{L}_{\noti})\geq 4^{-1}$ and satisfies a MSLI with constant $\alpha_{\noti} = \Omega(e^{-\beta})$.
\end{claim}

The \textit{unique} fixed point of $\mathcal{L}_{NI}$ is thus the tensor product state $\sigma_\beta = \otimes_i \sigma_\beta^i \propto e^{-\beta H_{\noti}}$. We defer a proof of \cref{claim:analysis-ni-gapped} to the next subsection. In the rest of this subsection, we show how to relate our Lindbladian $\mathcal{L}$ of \cref{equation:Lindbladian} (implicitly defined by the quantum circuit $C$), to $\mathcal{L}_{\noti}$, and moreover how to inherit its rapid mixing properties. 

\begin{claim}
    [A Convex Combination of Lindbladians]\label{claim:analysis-convex-combination} In a basis rotated by $C$, the Lindbladian $\mathcal{L}$ can be written as the convex combination
    \begin{equation}
        C^\dagger \mathcal{L}[C\cdot C^\dagger]C = q\cdot \mathcal{L}_{\noti}[\cdot] +(1-q) \cdot \mathcal{L}_{rest}[\cdot ],
    \end{equation}
    of two Lindbladians $\mathcal{L}_{\noti}, \mathcal{L}_{rest}$ which share the fixed point $\sigma_\beta = \otimes_i \sigma_\beta^i$. Moreover, the parameter $q = 4^{1-\ell}$ depends only on the lightcone size of $C$.  
\end{claim}

A proof of which we also defer to a future subsection. The convex combination claim above is the heart of our analysis, as it enables us to inherit the gap and mixing properties of $\mathcal{L}_{NI}$, without knowing properties of $\mathcal{L}_{rest}$ except for its (common) fixed point. To conclude this subsection, we present a proof of the MLSI of $\mathcal{L}:$

\begin{claim}
    [The Modified Log-Sobolev Inequality] \label{claim:analysis-mlsi} The Lindbladian $\mathcal{L}$ satisfies a MSLI with constant $\alpha \geq q\cdot \alpha_{NI} = \Omega(4^{-\ell}\cdot e^{-\beta}).$
\end{claim}

\begin{proof}

    [of \cref{claim:analysis-mlsi}] From \cref{claim:analysis-convex-combination}, we can write our Lindbladian $\mathcal{L}$ in a basis rotated by the circuit $C$ as a convex combination
    \begin{equation}
       \Tilde{\mathcal{L}} =  C^\dagger\mathcal{L}[C\cdot C^\dagger]C = q\cdot \mathcal{L}_{\noti}[\cdot] + (1-q)\cdot \mathcal{L}_{rest}[\cdot]
    \end{equation}

    \noindent  Since relative entropy is basis independent, proving a MLSI for $\Tilde{\mathcal{L}}$ similarly implies one for $\mathcal{L}$ with the same constant. To do so, we begin by expressing the ``entropy production rate" as a convex combination. 
    \begin{align}
      \ep_{\Tilde{\mathcal{L}}}(\rho) &=  \Tr[\Tilde{\mathcal{L}} [\rho](\log\rho-\log \sigma_\beta)]=  \\ &=  q \Tr[\mathcal{L}_{\noti}[\rho](\log\rho-\log \sigma_\beta)]+ (1-q) \Tr[\mathcal{L}_{rest}[\rho](\log\rho-\log \sigma_\beta)] 
    \end{align}
    To the first term on the RHS above, we can simply apply the MLSI for the non-interacting Lindbladian \cref{claim:analysis-ni-gapped}:
    \begin{equation}
        \Tr[\mathcal{L}_{\noti}[\rho](\log\rho-\log \sigma_\beta)] \leq -\alpha_\noti \cdot D(\rho||\sigma_\beta).
    \end{equation}

In turn, we claim that the second term on the RHS above is non-positive. Indeed, note that by \cref{claim:analysis-convex-combination}, $\sigma_\beta$ is a fixed point of $\mathcal{L}_{rest}$. The Data-processing inequality for the relative entropy then tells us that 
    \begin{equation}
        \Tr[\mathcal{L}_{rest}[\rho](\log\rho-\log \sigma_\beta)] = \frac{d}{dt}D(e^{t\mathcal{L}_{rest}}[\rho]||\sigma_\beta)\bigg|_{t=0} = \frac{d}{dt}D(e^{t\mathcal{L}_{rest}}[\rho]||e^{t\mathcal{L}_{rest}}[\sigma_\beta])\bigg|_{t=0} \leq 0.
    \end{equation}
Put together, we conclude $\ep_{\Tilde{\mathcal{L}}}(\rho) \leq -q\cdot \alpha_\noti\cdot D(\rho||\sigma_\beta)$.
\end{proof}

\subsection{The non-interacting Lindbladian is gapped (\texorpdfstring{\cref{claim:analysis-ni-gapped}}{})}

We dedicate this subsection to an analysis of the non-interacting Lindbladian $\mathcal{L}_{NI}$ (\cref{claim:analysis-ni-gapped}). 

\begin{lemma}\label{lemma:analysis_single_qubit_gapped}
    The spectral gap of the single-qubit Lindbladian $\mathcal{L}_{single}$ is $\Delta(\mathcal{L}_{single})\geq 4^{-1}$.
\end{lemma}

To understand this spectral gap, we revisit the (Hermitian) Discriminant super-operator $\mathcal{K}$ defined in \cref{definition:discriminant}. Recall, from \cref{lemma:discriminant_properties}, that (under detailed balance) this super-operator has the same eigenvalues of $\mathcal{L}$. In turn, to understand the spectral gap of $\mathcal{K}$, we vectorize this super-operator (on $2\times 2$ matrices) into an operator (a $4\times 4$ matrix).

\begin{equation}
    \mathcal{K}[\cdot] = \sum_j A_j [\cdot] B_j\rightarrow \textbf{K} = \sum_j A_j \otimes B_j^T.
\end{equation}

\begin{proof}
    To analyze the gap, we consider the discriminant $\mathcal{K}_{single}$ of the Lindbladian $\mathcal{L}_{single}$, and in particular its vectorization:
\begin{gather}
    \textbf{K}_{single} =  \sum_{a \in \mathcal{A}} \sum_{\nu \in [-n, n]}- \sqrt{\gamma(\nu)\gamma(-\nu)} \cdot A^a_\nu\otimes (A^a_\nu)^* +\frac{\gamma(\nu)}{2} \bigg((A^a_\nu)^\dagger A^a_\nu\otimes \mathbb{I} + \mathbb{I}\otimes (A^a_\nu)^T(A^a_\nu)^*\bigg)
\end{gather}
which is PSD, frustration free, and preserves the eigenvalues of $\mathcal{L}_{single}$ (up to a factor of $-1$) \cref{lemma:discriminant_properties}. Moreover, via detailed balance, the purified Gibbs state $\ket{\sqrt{\sigma_\beta}} \propto \ket{00}+e^{-\beta /2}\ket{11}$ is a ground state of $\textbf{K}_{single}$. Since the jump operators are single qubit Pauli operators $\{\mathbb{I}, X, Y, Z\}$, they can be written in the energy basis as
\begin{gather}
        A^{\mathbb{I}}_0 \propto \mathbb{I}\quad \text{and}\quad A^{Z}_0 \propto Z_i \quad\text{and}\quad A^{X}_{1} = (-i) A^{ Y}_{1}\propto  \begin{bmatrix} 0 & 0 \\ 1 & 0 \end{bmatrix}
    \end{gather}
and such that the conjugates can be inferred from the identity $A^a_{\nu} = A^{a\dagger}_{-\nu}$. The $4\times 4$ vectorized discriminant can therefore be written as

\begin{equation}
    \textbf{K}_{single} = \frac{1}{2} \begin{bmatrix}
            \gamma(1)  & 0 & 0 & - \sqrt{\gamma(1)\gamma(-1)} \\ 
             0 & \frac{\gamma(1) + \gamma(-1)}{2} + \gamma(0) & 0 & 0 \\
             0 & 0 & \frac{\gamma(1) + \gamma(-1)}{2} + \gamma(0) & 0\\
             - \sqrt{\gamma(1)\gamma(-1)}& 0 & 0 & \gamma(-1)
        \end{bmatrix}
\end{equation}
which we identify to be frustration free and have spectral gap $\frac{\gamma(1)}{2}\cdot \min(1+e^{\beta }, \frac{1+e^{\beta }}{2} + \frac{\gamma(0)}{\gamma(1)}) = \frac{\gamma(1)}{2}\geq \frac{1}{4}$ under Glauber Dynamics, where $\gamma(\nu) = (1+e^{-\beta \nu})^{-1}$.

\end{proof}

The positivity of the spectral gap can be used to show a complete MLSI, as shown by \cite{gao2022complete}. This conversion comes at the cost of factors of the local dimension of the Lindbladian - which in the case of $\mathcal{L}_{single}$, is just $2$.

\begin{theorem}
    [CMLSI from the Spectral Gap, \cite{gao2022complete} Theorem 4.3]\label{theorem:cmlsi_from_gap} Suppose a Lindbladian $\mathcal{G}$, acting on a $D$-dimensional Hilbert space, is GNS-symmetric w.r.t a fixed state $\sigma>0$.  Then, it satisfies a CMLSI with constant 
    \begin{equation}
        \alpha_{c} \geq \Delta(\mathcal{G}) \cdot \frac{\|\sigma^{-1}\|^{-1}}{D^2}.
    \end{equation}
\end{theorem}

In this manner, $\mathcal{L}_{single}$ satisfies a CMLSI with constant $\alpha_{single} = \frac{1}{16\cdot (1+e^{\beta})}$.\footnote{As remarked by an anonymous reviewer, it may be possible to remove the $\beta$ dependence, since the Lindbladian is non-interacting, following the results of \cite{MullerHermes2015EntropyPO, MullerHermes2015RelativeEC, Capel2018QuantumCR} or \cite{Beigi2018QuantumRH}.} We are now in a position to prove the MLSI for $\mathcal{L}_{NI}$.

\begin{proof}

    [of \cref{claim:analysis-ni-gapped}] We begin by leveraging the Complete MLSI, on the local Lindbladians (\cref{theorem:cmlsi_from_gap} and \cref{lemma:analysis_single_qubit_gapped}):
\begin{align}
\Tr[\mathcal{L}_{NI}[\rho] (\log(\rho) - \log(\sigma_\beta) )] &= \sum_{i\in [n]} \Tr[\mathcal{L}_{single}^i[\rho] (\log(\rho) - \log(\sigma_\beta) )]\\ &= \sum_{i\in [n]}  \Tr[\mathcal{L}_{single}^i[\rho] (\log(\rho) - \log(\mathcal{E}_i[\rho]) )] \\ & \leq -\alpha_{single} \sum_{i\in [n]}  D(\rho||\mathcal{E}_i[\rho]) 
\end{align}
Where $\mathcal{E}_i$ is the conditional expectation of the $i$th semigroup $e^{t\mathcal{L}_{single}^i}$. Next, we leverage the strong subadditivity of non-interacting conditional expectations \cite{gao2022complete} (eq. 5), see also \cite{Petz1991OnCP}:
\begin{equation}
   \sum_{i\in [n]}  D(\rho||\mathcal{E}_i[\rho]) \geq D(\rho||\prod_i \mathcal{E}_i[\rho]).
\end{equation}

\noindent To conclude, observe that for any $\rho$, the collection of conditional expectation $\prod_{i\in [n]} \mathcal{E}_i[\rho] = \sigma_\beta$ maps to the unique stationary state.
\end{proof}

\subsection{The convex combination claim (\texorpdfstring{\cref{claim:analysis-convex-combination}}{})}

Recall that the parent Hamiltonian $H$ has solvable eigenstates given by $\{C\ket{x}:x\in \{0, 1\}^n\}$, with energies given by the Hamming weight $|x|\in [n]$.

\begin{proof}

We begin by explicitly writing down each jump operator in the frequency basis:
    \begin{align}
        A_\nu^a = \sum_k \Pi_{k+\mu} A^a\Pi_k &= C\bigg( \sum_{k\in [n]} \sum_{\substack{|x| = k \\ |y| = k+\nu}} \ket{y}\bra{x} \cdot \bra{y}C^\dagger A^a C\ket{x}\bigg)C^\dagger \\
        &=: C\bigg( \sum_{k\in [n]} \sum_{\substack{|x| = k \\ |y| = k+\nu}} \ket{y}\bra{x} \cdot N_{x, y}^a\bigg)C^\dagger 
    \end{align}

    \noindent For conciseness, we have denoted the coefficient by $N_{x, y}^a= \bra{y}C^\dagger A^a C\ket{x}$. Our Lindbladian $\mathcal{L}$ of \cref{equation:Lindbladian} can thus be written in a basis rotated by the circuit $C$, in terms of the second moment of these coefficients:
    \begin{gather}
        C^\dagger\mathcal{L}[C\rho C^\dagger]C = \sum_{\nu} \gamma_{\nu} \cdot \sum_a   \bigg(\sum_{k, k'\in [n]} \sum_{\substack{|x| = k \\ |y| = k+\nu}} \sum_{\substack{|x'| = k' \\ |y'| = k'+\nu}}   N_{x, y}^a \cdot (N_{x', y'}^a)^* \cdot  \ket{y}\bra{x} \rho \ket{x'} \bra{y'} \\
        - \frac{1}{2} \sum_{k\in [n]} \sum_{\substack{|x|, |x'| = k \\ |y| = k+\nu}} N_{x, y}^a \cdot (N_{x', y}^a)^* \cdot \bigg\{ \ket{x'}\bra{x}, \rho\bigg\}\bigg).
    \end{gather}

    For $i\in [n]$, consider the subset of jump operators $\mathcal{A}_i$, centered around the $i$-th lightcone $\mathsf{L}_i$:
    \begin{align}
    \mathcal{A}_i = 2^{-\ell}\cdot \big\{P_{\mathsf{L}_i}\otimes \mathbb{I}_{[n]\setminus \mathsf{L}_i}:  P\in \mathcal{P}_{\ell}\big\} 
    \end{align}
    By definition, these subsets are disjoint, and form a partition $\cup_i \mathcal{A}_i =  \mathcal{A}$. We claim that we can rotate the jump operators in each subset, by substituting
    \begin{equation}
        \mathcal{A}_i\rightarrow \mathcal{A}_i'= U_i\mathcal{A}_i U_i^\dagger
    \end{equation}

    \noindent for an \textit{arbitrary} choice of unitary $U_i$ of support contained in $\mathsf{L}_i$, while keeping the Lindbladian $\mathcal{L}$ invariant. 
    Essentially, this is because the Lindbladian is only defined by the second moments of the jump operators, and that the second moment of random Pauli operators is Haar random (via the 1-design property) and thus invariant under unitary conjugation:
    \begin{align}
        \sum_{a\in \mathcal{A}_i} A^a[\cdot]A^{a} = \frac{1}{2^{\ell}}\tr_{\mathsf{L}_i}[\cdot] = \sum_{a\in \mathcal{A}_i} U_i^{\dagger}A^aU_i[\cdot]U_i^{\dagger}A^{a}U_i\quad \text{for any $U_i$ supported on $\mathsf{L}_i$}. 
    \end{align}
    Indeed, for every choice of basis elements $x, x', y, y'\in \{0, 1\}^n$, the pre-factor
    \begin{align}
       \sum_{a\in \mathcal{A}_i}N_{x, y}^a\cdot  (N_{x', y'}^a)^* &=\sum_{a\in \mathcal{A}_i} \bra{y}C^\dagger A_a C\ket{x} \bra{x'} C^\dagger A_a C\ket{y'}   \\ &=\sum_{a\in \mathcal{A}_i} \bra{y}C^\dagger U_i A_a U_i^{\dagger} C\ket{x} \bra{x'} C^\dagger U_i A_a U_i^\dagger C\ket{y'} = \sum_{a'\in \mathcal{A}_i'}N_{x, y}^{a'}\cdot  (N_{x', y'}^{a'})^* 
    \end{align}
    \noindent is preserved, whether in $\mathcal{A}_i$ or $\mathcal{A}_i'$. 
    
    Finally, let $U_i$ be the gates in $C$ contained in the lightcone of the $i$th qubit. The sum over jump operators $a\in \mathcal{A}$ can be written as an expectation over random $\ell$-qubit Paulis, $P\in\mathcal{P}_\ell$, and then an expectation over the center $i\in [n]$ in which to place $P$. With probability $4/ 4^{\ell}$, $P$ is a single qubit Pauli $P_i$ centered at $i$. Moreover, for a single-qubit Pauli $P_i$ centered at $i$, the choice of $U_i$\textit{ exactly cancels with the circuit} $C^\dagger$:
    \begin{align}
        &\bra{y}C^\dagger U_i\bigg(P_i \otimes \mathbb{I}_{[n]\setminus \{i\}}\bigg) U_i^\dagger C\ket{x} = \bra{y} \bigg(P_i \otimes \mathbb{I}_{n\setminus \{i\}}\bigg)\ket{x} \text{, }\\
        &\text{for each}\quad x, y \in \{0, 1\}^n \quad\text{and}\quad P_i\in \{\mathbb{I}, X, Y, Z\}_i.
    \end{align}

    We note that these are precisely the jump operators we expect in the non-interacting case $\mathcal{L}_{\noti}$, where the circuit is replaced by the trivial circuit $C=\mathbb{I}$. In this manner, we conclude that the rotated Lindbladian can be written a convex combination:
    \begin{equation}
         C^\dagger \mathcal{L}[C\cdot C^\dagger]C = 2^{2(1-\ell)}\cdot \mathcal{L}_{\noti}[\cdot] +(1-2^{2(1-\ell)}) \cdot \mathcal{L}_{rest}[\cdot ],
    \end{equation}

    \noindent where both $\mathcal{L}_{\noti}[\cdot]$, $\mathcal{L}_{rest}[\cdot ]$ are Davies' generators defined on disjoint sets of jump operators. Both of them satisfy detailed balance and share the Gibbs state as the stationary state.
\end{proof}

\section{Circuit Implementation of the Dissipative Lindbladian}
\label{section:implementation}

The main claim of this section is an efficient implementation of the Lindbladian time-evolution using a quantum circuit. Put together with our bound on the mixing time of our Lindbladians, this all but concludes the proof of the preparation of Gibbs states of the parent Hamiltonians of quantum circuits. 

\begin{lemma}[Dissipative Lindbladian Implementation]\label{lemma:dissipative_implementation} Fix parameters $t \geq  1$ and $\epsilon\leq \frac{1}{2}$. Let $\mathcal{L}$ denote the Lindbladian of \cref{equation:Lindbladian}, defined by a quantum circuit on $n$ qubits, of lightcone size $\ell$ and depth $d$. Then, we can simulate the map $e^{t\mathcal{L}}$ to error $\epsilon$ in diamond norm using a quantum algorithm of depth $O(t\cdot n\cdot 2^{2\ell}\cdot 2^d\cdot  \poly(\ell, \log n, \log \frac{1}{\epsilon}, \log t))$.
\end{lemma}

We dedicate \cref{subsection:prelim} to presenting the required background results on implementing Lindbladian evolution using quantum circuits. In the ensuing section \cref{subsection:optimizing}, we discuss optimizations both particular to our systems, and generic, to the runtime of our algorithms. 

\subsection{Preliminaries on simulating Lindbladian evolution}
\label{subsection:prelim}

Our implementation of the map $e^{t\mathcal{L}}$ follows the framework of \cite{Chen2023QuantumTS}, in reducing the task to constructing a block-encoding of the Lindblad operators. To implement their scheme it is suitable to renormalize the time-scale and Lindblad operators $\{L_j\}_{j\in \mathcal{A}}$ such that the resulting Lindbladian has norm 1:
\begin{equation}
    t\rightarrow t \cdot \|\sum_j L_j^\dagger L_j\|  \text{ and } L_j\rightarrow L_j \cdot \|\sum_j L_j^\dagger L_j\|^{-1/2}
\end{equation}

\noindent Given the choice of jump operators from \cref{equation:Lindbladian}, under this normalization we have $t\rightarrow  n\cdot t$.

\begin{definition}[Unitary block encoding for Lindblad Operators{~\cite[Definition I.2]{Chen2023QuantumTS}}]\label{def:block-encoding}
    Given a purely irreversible Lindbladian determined by the Lindblad operators $\{L_j\}_{j\in \mathcal{A}}$, a unitary $U$ is said to be block-encoding of the Lindblad operators if
    \begin{equation}
        \big(\bra{0}^b\otimes \mathbb{I}\big) U \big(\ket{0}^c\otimes \mathbb{I}\big) = \sum_{a\in \mathcal{A}} \ket{a}\otimes L_j \text{ for }b, c\in \mathbb{N}
    \end{equation}
\end{definition}

Given a black-box circuit corresponding to a block-encoding of $\mathcal{L}$, the following theorem stipulates that one can simulate the corresponding Lindbladian evolution for time $t$ using just $\Tilde{O}(t)$ invocations of the black-box:

\begin{theorem}
    [Theorem I.2, \cite{Chen2023QuantumTS}]\label{theorem:Lindblad_implementation_bb} Suppose $U$ is a unitary block-encoding of the Lindbladian $\mathcal{L}$ as in \cref{def:block-encoding}. Let time $t \geq  1$ and error $\epsilon\leq \frac{1}{2}$, then we can simulate the map $e^{t\mathcal{L}}$ to error $\epsilon$ in diamond norm using 
    \begin{enumerate}
        \item $O((c+\log \frac{t}{\epsilon})\log \frac{t}{\epsilon})$ resettable ancilla qubits,
        \item $\Tilde{O}(t)$ controlled uses of $U$ and $U^\dagger$, and 
        \item $\Tilde{O}(t+c)$ other 2-qubit gates.
    \end{enumerate}
\end{theorem}

We remark that since our Hamiltonian has a integer spectra $[n]$, one can exactly implement the projection of the jump operators $\{A^a\}$ onto the energy eigenbasis by performing an operator Fourier transform with uniform weights:

\begin{equation}
    A_\nu^a \propto \sum_{\Bar{t}\in S_{\pi/n}} e^{i\nu \bar{t}} e^{i H\bar{t}} A^a e^{-iH\bar{t}} \text{ where } S_{\pi/n} = \frac{\pi}{n}\cdot \{ -n, -(n-1), \cdots, -1, 0, 1, \cdots, n\}
\end{equation}
In this setting, we can now apply a lemma on the efficient implementation of block-encodings from \cite{Chen2023QuantumTS}, simplified to the context of integer spectra Hamiltonians.

\begin{lemma}[Lemma I.1, \cite{Chen2023QuantumTS}]\label{lemma:Lindblad_block_encoding} In the setting of \cref{theorem:Lindblad_implementation_bb}, a unitary block encoding for the Lindblad operators corresponding to a Hamiltonian $H$ of integer spectra $[n]$ can be created using $O(n+\log |A|)$ ancilla qubits, as well as one query to

\begin{enumerate}
    \item The controlled Hamiltonian simulation: $\sum_{\bar{t}\in S_{\pi / n}} \ketbra{\bar{t}} \otimes e^{\pm i H \bar{t}}$,
    \item A block-encoding of the jump operators: $\sum_{a\in \mathcal{A}} \ket{a}\otimes A^a$,
    \item $O(\log n)$ qubit Quantum Fourier transform: $\ket{\bar{t}}\rightarrow (2n)^{-1/2}\sum_{\omega\in [-n, \cdots n]} e^{-i\omega \bar{t}}\ket{\omega}$
    \item And a controlled filter for the Boltzmann factors:
    \begin{equation}
        W = \sum_{\omega\in [-n, \cdots, n]}\begin{bmatrix}
            \sqrt{\gamma(\omega)} & - \sqrt{1-\gamma(\omega)} \\ 
            \sqrt{1-\gamma(\omega)} & \sqrt{\gamma(\omega)}
        \end{bmatrix}\otimes \ketbra{\omega}
    \end{equation}
\end{enumerate}
\end{lemma}

\subsection{Optimizing the circuit implementation}
\label{subsection:optimizing}

In light of \cref{theorem:Lindblad_implementation_bb} and \cref{lemma:Lindblad_block_encoding}, in what remains of this section, we describe how to implement the controlled Hamiltonian simulation (\cref{claim:hamiltonian-simulation}), the block-encoding of the jump operators (\cref{claim:jumps}), and the controlled Boltzmann filter (\cref{claim:boltz_filter}), in circuit depth $O(4^\ell\cdot 2^d\cdot \poly(\log n, \log \frac{1}{\epsilon}, \ell))$. While the first two optimizations are particular to our family of Hamiltonians, the latter may find independent application to the framework of \cite{Chen2023QuantumTS}.

We begin with a simple lemma which ``colors" the interaction graph of the Hamiltonian, partitioning the interactions into disjoint subsets $S_1, S_2\cdots S_\Delta \subset [n]$ such that no two terms $h_i, h_j$ of the same subset have overlapping support.

\begin{lemma}\label{lemma:coloring}
    Any parent Hamiltonian $H\in \mathscr{H}$ defined by a quantum circuit of depth $d$ and lightcone size $\ell$ can be $\Delta$-colored with $\Delta \leq \ell \cdot 2^d + 1$ colors. 
\end{lemma}

\begin{proof}
    Two interactions $h_i, h_j$ overlap at a qubit only if their lightcones intersect in the underlying circuit $C$, which determines $H$. Let $\ell_r$ denote the maximum ``reverse lightcone" size of the circuit, that is, the maximum number of qubits which have a given qubit in their lightcone. Since the Hamiltonian is at most $\ell$-local, any interaction $h_i$ overlaps with at most $\ell\cdot \ell_r$ other terms, which in turn tells us the interactions can then be partitioned using $\Delta\leq \ell\cdot \ell_r+1$ different colors. If the depth of the circuit $C$ as measured by layers of 2-qubit gates is $d$, then $\ell_r\leq \min(n, 2^d)$. 
\end{proof}

\begin{claim}\label{claim:hamiltonian-simulation}
    The controlled time-evolution of an $n$ qubit parent Hamiltonian of a quantum circuit with lightcone size $\ell$, can be implemented using a quantum circuit of depth $O(4^\ell \cdot \Delta\cdot \log n)$ and size $O(4^\ell \cdot \Delta \cdot n\cdot\log n)$.
\end{claim}

At a high level, the circuit of \cref{claim:hamiltonian-simulation} partitions the terms of the (commuting) Hamiltonian into disjoint subsets of non-overlapping terms, which can be implemented in parallel. However, since we need to implement the \textit{controlled} Hamiltonian simulation, all of these Hamiltonian terms need to act conditioned on the time-register, which is a sequential bottleneck to the circuit depth. In order to further compress the depth, we parallelized the access to the time-register by encoding it into a GHZ state.

\begin{proof}

    Since the Hamiltonian is commuting, let us restrict our attention to a fixed subset $S$ in the partition guaranteed by \cref{lemma:coloring}. It suffices to prove how to implement the controlled time evolution of each subset of non-overlapping terms $H_S = \sum_{i\in S} h_{i}$. For this purpose, we begin by parallelizing the access to the clock register: $\ket{\Bar{t}}\rightarrow \ket{\Bar{t}}^{\otimes n}$, using a $O(\log n)$ depth circuit of CNOT gates, of size $O(n\log n)$.
    
    Next, controlled on the $j$th clock register, we apply the time evolution of the $j$th interaction. Although this gate acts on $O(\log n) + \ell$ qubits, it can be implemented via a sequence of $O(\log n)$ gates acting only on $\ell+1$ qubits, by applying a binary expansion of the time register $\bar{t} = \frac{\pi}{n}\cdot \sum_k 2^k \bar{t}_k$:

    \begin{equation}
        \sum_{\bar{t}}\ketbra{\bar{t}}\otimes e^{i\bar{t}h_j} =  \prod_{k} \mathbb{I}_{\setminus k}\otimes  \sum_{\bar{t}_k \in \{0, 1\}} \ketbra{\bar{t}_k}\otimes \exp[i  \frac{\pi}{n} \cdot  2^k \bar{t}_k \cdot  h_j]
    \end{equation}

    In turn, each unitary on $\ell+1$ qubits can generically be implemented in $O(4^\ell)$ size and depth. After all the colors have concluded, we revert the copies of the clock register. 
\end{proof}

\begin{claim}\label{claim:jumps}
    A block encoding of the jump operators can be implemented using a quantum circuit of depth $O(\Delta\cdot (\ell + \log n))$ and size $O(n\cdot \Delta\cdot (\ell + \log n))$.
\end{claim}
\begin{proof}
    There are $|A| = n\cdot 4^\ell$ jump operators, which we represent by indexing them using a pair $a=(i, P)$ in terms of the center of its support $i\in [n]$, using $\log n$ qubits, as well as the $\ell$-local Pauli $P$, using $2\cdot \ell$ qubits. Next we proceed using similar techniques to \cref{claim:hamiltonian-simulation}. We begin by partitioning the jump operators into $\Delta$ disjoint subsets using \cref{lemma:coloring}, where any two jump operators in the same subset either have the same center, or do not intersect. Our implementation proceeds by addressing each color $c\in [\Delta]$ independently.
    
    First, we create copies of the control register $\ket{a}\rightarrow \ket{a}^{\otimes n}$, to parallelize access to it. Suppose all the jump operators centered at $j$ have been colored $c$. Our goal is to coherently apply all the controlled jump operators of the form $(j, P)$ by acting only on the support (centered at $j$) and the $j$th control register $\ket{a = (i, Q)}_j$. For this purpose, we first check whether $i=j$, and controlled on the one-qubit outcome, we apply the Pauli $Q$. The check can be implemented using $O(\log n)$ size and depth, and the controlled Pauli in $O(\ell)$ size and depth. We conclude by inverting the checking and copying steps.
\end{proof}

\begin{claim}\label{claim:boltz_filter}
    The controlled filter $W$ can be implemented up to error $\epsilon$ in spectral norm using a circuit of size $O(\polylog(\frac{n}{\epsilon}))$ 2-qubit gates.
\end{claim}

\begin{proof}

    Let us denote $n_\delta = \beta^{-1} \log \frac{1}{\delta}$. Then, the Glauber dynamics weight $\gamma(\nu) = (1+e^{-\beta \nu})^{-1}$ satisfies
    \begin{align}
        \gamma(\nu) &\leq \delta\quad \quad\text{if}\quad \nu \leq - n_\delta, \\
        \text{ and } \gamma(\nu) &\geq 1- \delta \quad\text{if}\quad \nu \geq n_\delta.
    \end{align}
    \noindent We claim that the $W$ gate can be replaced by a truncation $W_\delta$, 
    \begin{equation}
        W_\delta = \sum_{\omega\in [-n, \cdots, n]}\begin{bmatrix}
            \sqrt{\tilde{\gamma}(\omega)} & - \sqrt{1-\tilde{\gamma}(\omega)} \\ 
            \sqrt{1-\tilde{\gamma}(\omega)} & \sqrt{\tilde{\gamma}(\omega)}
        \end{bmatrix}\otimes \ketbra{\omega},\quad  \tilde{\gamma} (\omega) := \begin{cases}
            \gamma(\omega) &\text{ if } \omega \in [-n_\delta, n_\delta] \\
            1 &\text{ if } \omega > n_\delta \\
            0 &\text{ if } \omega < -n_\delta
        \end{cases}.
    \end{equation}

    Indeed, the truncation error is controlled by
    \begin{align}
    \|W -W_\delta\|
        &\leq  \sum_{j\in [-n, -n_\delta]} \| \begin{bmatrix}
            \sqrt{\gamma(\omega)} &  1 - \sqrt{1-\gamma(\omega)} \\ 
            \sqrt{1-\gamma(\omega)} - 1 & \sqrt{\gamma(\omega)}
        \end{bmatrix} \| \\ &+ \sum_{j\in [n_\delta, n]} \| \begin{bmatrix}
            \sqrt{\gamma(\omega)} - 1 &   - \sqrt{1-\gamma(\omega)} \\ 
            \sqrt{1-\gamma(\omega)}  & \sqrt{\gamma(\omega)} - 1
        \end{bmatrix} \|\\
        &\leq 2n\cdot (2\sqrt{\delta}+2\delta)\leq 8n\cdot \sqrt{\delta},
    \end{align}
    
    \noindent where the last line uses that $1-\sqrt{1-x}\leq x$ when $x\in [0, 1/2]$.
    
    It only remains now analyze the gate complexity of implementing $W_\delta$. Following \cite{Chen2023QuantumTS} (pg. 25, footnote 33), the $W_{\delta}$ filter for the Glauber weight between $[-n_\delta, n_\delta]$ can be implemented using the QSVT up to error $\epsilon$ using $\Tilde{O}((1+\beta n_\delta)\polylog \frac{1}{\epsilon})$ 2-qubit gates. With the choice $\delta = O(\frac{\epsilon}{n})$, we arrive at the advertised bounds by combining with the trivial cases $\omega \notin[-n_{\delta},n_{\delta}]$.
\end{proof}

We remark that this error in spectral norm between unitaries is equivalent  to the channel diamond norm distance, up to a constant: $\|U-V\|_\diamond \leq 2\cdot \|U-V\|$. 

Put together, \cref{claim:hamiltonian-simulation}, \cref{claim:jumps} and \cref{claim:boltz_filter} imply \cref{lemma:dissipative_implementation}.

\section{Low-Depth State Preparation on Lattices}
\label{sec:latticeprep}

We dedicate this section to a proof of \cref{theorem:gibbs-prep-lattices}, on the preparation of Gibbs states of parent Hamiltonians of quantum circuits defined on lattices.

\begin{theorem}\label{theorem:gibbs-prep-lattices}
    Fix an inverse-temperature $\beta >0 $, and let $H$ be the parent Hamiltonian of an $n$ qubit, depth $d$, quantum circuit comprised of 2-qubit nearest neighbor gates in $D$ dimensions. Then, there exists a quantum algorithm which prepares the Gibbs state of $H$ up to error $\epsilon$ in depth $2^{O(d^D)}\cdot \polylog\frac{n}{\epsilon}$.
\end{theorem}

The algorithm of \cref{theorem:gibbs-prep-lattices} based on that of \cite{Brando2016FiniteCL}, who showed that if the Gibbs state satisfies two structural decay-of-correlations properties, then it can be efficiently prepared by a quantum computer. In particular, a decay of the conditional mutual information (CMI) and a certain clustering of correlations. In this section we show that parent Hamiltonians of low depth circuits satisfy strengthened versions of both of these properties, giving rise to a state preparation circuit of nearly constant depth.

\subsection{The Markov Property}

The condition mutual information is an information-theoretic measure of the correlations in a tripartite state $\rho$:
\begin{equation}
    I(A:C|B)_\rho = S(AB)_\rho + S(BC)_\rho  -  S(ABC)_\rho -  S(B)_\rho,
\end{equation}

\noindent where $S(A)_\rho = -\Tr \rho_A\log \rho_A$ is the von Neumann entropy. Roughly, this quantity captures the mutual information between systems $A$ and $C$, conditioned on system $B$.

The quantum states with vanishing conditional mutual information are known as quantum Markov chains, which is equivalent to the existence of a local recovery map that reconstructs $\rho$ from just its subsystems. The following results also handle the cases with errors. 

\begin{theorem}[\cite{Fawzi2014QuantumCM}]
        For any state $\rho$ on $ABC$, there exists a quantum channel $\mathcal{R}_{B\rightarrow BC}$ such that 
        \begin{equation}
            I(A:C|B)_\rho \geq \frac{1}{4\ln 2} \|\rho - [\mathbb{I}_A\otimes \mathcal{R}_{B\rightarrow BC}](\rho_{AB})\|_1^2.
        \end{equation}
    \end{theorem}

\noindent In the special case of $I(A:C|B)_\rho = 0$, the recovery map is explicit and can be chosen to be the Petz (or transpose) map \cite{Petz1991OnCP, Junge2015UniversalRF}:
\begin{equation}\label{equation:petz}
    \mathcal{R}_{B\rightarrow BC}[X] = \rho_{BC}^{1/2}\big(\rho_{B}^{-1/2}X\rho_{B}^{-1/2}\otimes \mathbb{I}_C\big)\rho_{BC}^{1/2}
\end{equation}

In the context of this work, we require characterization of the CMI of the Gibbs states of commuting Hamiltonians. Fix a local commuting Hamiltonian $H$ defined on a set of vertices (the qubits) and hyper-edges (the interactions). Given three disjoint subsets $A, B, C$ of the qubits of $H$, we say that $B$ \textit{shields} $A$ from $C$ if all paths on the hyper-edges from $A$ to $C$ must pass through $B$. For instance, the boundary of a region $C$, comprised of the qubits which share an interaction with $C$, shield $C$ from the rest of the lattice/qubits. 

The Quantum Hammersley-Clifford Theorem states that the Gibbs states of commuting local Hamiltonians form quantum Markov Chains, for \textit{shielding} tripartitions. 

\begin{fact}
    [The Quantum Hammersley-Clifford Theorem, \cite{Brown2012QuantumMN}]\label{fact:qhc} Let $H$ denote a commuting local Hamiltonian and $A, B, C$ three disjoint subsets of the qubits of $H$ such that $B$ shields $A$ from $C$. Then, 
    \begin{equation}
        I(A:C|B)_{\rho} = 0, \text{ where } \rho = e^{-H}/\Tr e^{-H}.
    \end{equation}
\end{fact}

We remark that parent Hamiltonians of depth $d$ quantum circuits have local interactions of range $2\cdot d$. In this manner, the boundary of any region $C$ on the lattice is comprised of all the qubits at distance $\leq 2\cdot d$ from $C$, which aren't already in $C$.

\subsection{Local Indistinguishability}

Given a local Hamiltonian $H$ defined on a lattice $\Lambda$, and a subset of said lattice $X\subset \Lambda$, we refer to the subsystem Hamiltonian $H_X$ as all the interactions contained \textit{entirely} in $X$\footnote{i.e., the interactions act trivially outside subsystem $X$}:
\begin{equation}
    H_X = \sum_{e\subset X} h_e.
\end{equation}

\noindent Henceforth let us denote as $\rho \propto e^{-\beta H}$ the Gibbs state of $H$, and $\rho^{X} \propto e^{-\beta H_X}$ the Gibbs state of the subsystem Hamiltonian.

\begin{definition}[Local indistinguishability]
    Let $H$ be a local Hamiltonian defined on a lattice $\Lambda$, $\beta$ an inverse temperature, and $ABC=X\subset \Lambda$ be a partition of a subset of the lattice. Then the Gibbs state of $H$ is said to satisfy \emph{local indistinguishability} on $ABC$ if
    \begin{equation}
        \Tr_{BC}\rho^X = \Tr_B \rho^{AB}.
    \end{equation}
\end{definition}

    That is, the thermal expectation of local observables for a global Hamiltonian can be evaluated using only knowledge of a local Hamiltonian patch. This notion is closely related to the stability of gapped phases \cite{Michalakis2011StabilityOF}. \cite{Brando2016FiniteCL} show that an approximate version of the above follows if one assumes a clustering property for the correlations in the Gibbs state of $H$.

    In the context of this work, we claim that the structure of parent Hamiltonians of low depth circuits implies an exact local-indistinguishability property for their associated Gibbs states.

    \begin{claim}[Exact local indistinguishability for our parent Hamiltonians]\label{claim:exact-li}
        Let $C$ be a quantum circuit comprised of $d$ layers of nearest neighbor gates on a $D$ dimensional lattice $\Lambda$, and $H$ its associated parent Hamiltonian. Then, the Gibbs state of $H$ satisfies local indistinguishability for arbitrary $ABC = X\subset \Lambda$ such that the distance $d(A, C)\geq 4\cdot d+1$.
    \end{claim}

    Here we define the distance between subsets of the lattice $d(A, C) = \min_{i\in A, j\in C}d(i,j)$, defined by the path-length over edges of the lattice. 

    \begin{proof}
        We claim that the unitary $U$ comprised of all the gates in the reverse-lightcone of $B$ i.e, the union of all the lightcones $L_i$ which are entirely contained in $B$, disentangles $A$ from $C$ in the Gibbs state:
        \begin{equation}
            (U^\dagger\otimes \mathbb{I})\rho^X(U\otimes \mathbb{I}) = \sigma_{AB_1}\otimes \gamma_{B_2 C}.
        \end{equation}
        \noindent This decoupling implies the desired property. To begin, note that the depth $d$ of the circuit $C$ implies a bound of $\leq 2d$ on the range of interaction of the terms in the Hamiltonian $H$. In this manner, if $d(A, C)\geq 2\cdot d+1$ then $B$ screens $A$ from $C$ in the subsystem Hamiltonian $H_X$. 

        Suppose we partition the qubits of $B$ into $B_A$, all qubits of $B$ at distance $\leq d$ from $A$, $B_C$, those at distance $\leq d$ from $C$, and $B_- = B\setminus (B_A\cup B_C)$. Since $d(A, C)\geq 4d+1$, no two terms $h_i, h_j$, $i\in A\cup B_A$ and $j\in C\cup B_C$ have intersecting support. Moreover, for $k\in B_-$ such that the lightcone $L_k\subset B$, 
        \begin{equation}
            U^\dagger h_k U = U^\dagger C(\ketbra{1}_k\otimes \mathbb{I}_{[n]\setminus \{k\}})C^\dagger U = \ketbra{1}_k,
        \end{equation}

        \noindent and therefore other $U^\dagger h_i U$ has support on any element of $B_-$. Thus, the unitary $U$ partitions the Hamiltonian into non-interacting components:
        \begin{equation}
            U^\dagger H_XU = H_{AB_A} + H_{B_- B_C C}.
        \end{equation}
        To conclude, we observe that $U^\dagger \rho^XU$ is the Gibbs state of $U^\dagger H_XU$.
     \end{proof}

    We remark that the claim above can alternatively be seen as a consequence of both clustering of correlations and the quantum belief propagation equations with 0 error \cite{Brando2016FiniteCL, Hastings2007QuantumBP}.

\subsection{The Algorithm}

We can now invoke a much simplified version of the algorithm of \cite{Brando2016FiniteCL}, on the preparation of Gibbs states of Hamiltonians of $D$-dimensional lattices. Their algorithm is based on the operation of recovery maps on finite-sized regions, namely up to a certain correlation length-scale $\ell_0$. To instantiate their algorithm, it remains to show how to perform these recovery maps. 

Fortunately, we can appeal to the explicit structure of the Petz map of \cref{equation:petz}. Via a Stinespring dilation\footnote{See e.g. Chapter 2 of \cite{Watrous2018TheTO}}, each such recovery channel can be purified into a quantum circuit on a doubled Hilbert space. In turn, said circuit can be implemented in exponential time in the volume (the number of qubits) of the region, up to factors of $\polylog\frac{n}{\epsilon}$. See also \cite{gilyen2022quantumPetz} for a more general discussion on implementing Petz recovery maps. 

\begin{theorem} [\cite{Brando2016FiniteCL},  Theorem 6, Simplified] Fix $\beta > 0$, let $H$ be a bounded local Hamiltonian defined on $D$ dimensional lattice $\Lambda$, and let $\rho$ be its Gibbs state at inverse-temperature $\beta$. If $\rho$ is both locally indistinguishable and exactly Markov for all tripartitions $ABC=X\subset \Lambda$ such that $d(A, C)\geq \ell_0$, then a purification of $\rho$ can be prepared up to error $\epsilon$ using a quantum circuit of depth $2^{O(\ell_0^D)}\cdot \polylog \frac{n}{\epsilon}$.
\end{theorem}

In full generality, \cite{Brando2016FiniteCL}'s result applies to Gibbs states which are approximately Markov and approximately locally indistinguishable. However, the typically exponentially-decaying tail in the approximation error implies quasi-polynomial-depth state preparation algorithms. By appealing to their exact counterparts, we arrive at a polylog depth state-preparation algorithm. To conclude \cref{theorem:gibbs-prep-lattices}, we note that $\ell_0 = O(d)$ from \cref{fact:qhc} and \cref{claim:exact-li}.

\section{The Input Noise Model and Gibbs States of Quantum Circuits}
\label{section:noise}

In this section, we show that the Gibbs states of parent Hamiltonians of quantum circuits correspond to noisy versions of the output of the quantum circuit, under a certain input noise model. To begin, let us recollect the noise model. Fix a noise rate $p\in (0, 1)$. The single-qubit bit-flip error channel consists of the superoperator
\begin{equation}\label{equation:bit-flip}
    \mathcal{D}_p(\sigma) = (1-p)\cdot \sigma + p\cdot X\sigma X.
\end{equation}

Given a quantum circuit $C$ on $n$ qubits, the input noise model consists of independent applications of the bit-flip error channel on the input wires of $C$. In particular, the mixed state given by the output of the noisy circuit is:
\begin{equation}
        \rho = C\bigg( \mathcal{D}_p(\ketbra{0})\bigg)^{\otimes n} C^\dagger
    \end{equation}

\noindent For a fixed $n$ qubit quantum circuit $C$, recall that we refer to the parent Hamiltonian of $C$ as 

\begin{equation}
    H_C = C\bigg(\sum_{i\in [n]} \ketbra{1}_i\otimes \mathbb{I}_{[n]\setminus i}\bigg)C^\dagger
\end{equation}

\begin{lemma}\label{lemma:gibbs_input_noise}
    Fix $\beta>0$, and let $H_C$ be the parent Hamiltonian of a quantum circuit $C$. The Gibbs state of $H_C$ at inverse-temperature $\beta$ is given by the output of the circuit $C$ under input level noise with probability $p = (1+e^{\beta})^{-1}:$
    \begin{equation}
        \rho_\beta =\frac{ e^{-\beta H_C}}{\Tr e^{-\beta H_C}} = C \bigg( \mathcal{D}_p(\ketbra{0})\bigg)^{\otimes n} C^\dagger
    \end{equation}
\end{lemma}

\begin{proof}
    It suffices to consider the Gibbs state $\sigma_\beta$ of the Hamiltonian $H = \sum_{i\in [n]} \ketbra{1}_i $, as $\rho_\beta = C\sigma_\beta C^{\dagger}$. Since $H$ is commuting, the partition function can be written as:
    \begin{gather}
        \Tr e^{-\beta H_C} = \Tr e^{-\beta H} = \sum_{x\in \{0, 1\}^n} \prod_i^n \bra{x_i} e^{-\beta \ketbra{1}_i}\ket{x_i} = \\ = \prod_i^n \sum_{x_i\in \{0, 1\}} \bra{x_i} e^{-\beta \ketbra{1}_i}\ket{x_i} = (1+e^{-\beta})^n.
    \end{gather}

    \noindent Therefore, the Gibbs state of $H$ can be expressed as the outcome of the depolarizing channel:
    \begin{equation}
        \sigma_\beta = (1+e^{-\beta})^{-n}\cdot e^{-\beta H} = \bigotimes_i^n \bigg( \frac{\ketbra{0}}{1+e^{-\beta}} + \frac{\ketbra{1}}{1+e^{\beta}}\bigg) = \bigg( \mathcal{D}_p(\ketbra{0})\bigg)^{\otimes n},
    \end{equation}

    \noindent with $p = (1+e^{\beta})^{-1}$.
    
\end{proof}

\section{Computational Complexity of Shallow IQP Sampling}
\label{section:appendix-hardness}

In recent years several architectures have been proposed for achieving a quantum speedup, based on quantum processes which resemble or are equivalent to the IQP Circuit Sampling task discussed in \cref{section:hardness}. The basis for these speedups is on standard complexity-theoretic conjectures, including the non-collapse of the Polynomial Hierarchy, often in addition to strong assumptions on the hardness of computing permanents or partition functions. We dedicate this section to a discussion on the background behind \cref{theorem:sampling}, as well as a comparison to related statements in the literature. 

To begin, let us recollect the circuit described in \cref{section:hardness}, comprised of a 2D cluster state and random phase gates \cite{BermejoVega2017ArchitecturesFQ, Gao2016QuantumSF}.

\begin{figure}[H]
\begin{tcolorbox}
\begin{enumerate}
    \item Prepare an $n$ qubit cluster state on a 2D rectangular lattice using a layer of Hadamard gates and 4 layers of CZ gates.

    \item Sample a random string $b \in [7]^n$, and apply powers of $T$ gates to each qubit:
    \begin{equation}
        \bigotimes_{i\in [n]}T^{b_i}\bigg(\prod_{<i, j>} CZ_{i, j} \ket{+}^{\otimes n}\bigg) = C_b\ket{0}^{\otimes n}, \text{ where }T = \begin{bmatrix}
            1 & 0\\
            0 & e^{i\pi/4}
        \end{bmatrix}.
    \end{equation}

    \item Finally, measure the output in the $X$ basis. 
\end{enumerate}
\end{tcolorbox}
\caption{A family of random IQP circuits, $\{C_b\}$.}
    \label{fig:iqp_circuits}
\end{figure}

If instead of \textit{random} powers of single-qubit $T$ gates, the powers were chosen adaptively given partial measurements of the circuit, this scheme would implement measurement-based quantum computation \cite{Broadbent2008UniversalBQ}. The universality of MBQC under adaptivity (or post-selection) implies the hardness of \textit{exactly} sampling from the output distribution, unless the polynomial hierarchy collapses to the third level \cite{Bremner2010ClassicalSO, Gao2016QuantumSF}. To reproduce their argument, universality implies

\begin{equation}
    \mathsf{PostIQP} \under{=}{\text{\cite{Bremner2010ClassicalSO}}} \mathsf{PostBQP} \under{=}{\text{\cite{Aaronson2004QuantumCP}}} \mathsf{PP}.
\end{equation}

\noindent If we now assume there existed a classical algorithm to exactly sample from arbitrary IQP circuits, that would imply $\mathsf{PP} = \mathsf{PostIQP}\subseteq \mathsf{PostBPP}$, which in turn gives us a collapse of the Polynomial Heirarchy (henceforth, $\mathsf{PH}$):
\begin{equation}
  \mathsf{PH}\under{=}{\text{Toda's Theorem}}  P^{\mathsf{PP}} \under{=}{\text{By assumption}} P^{\mathsf{PostBPP}}  = \Sigma_3.
\end{equation}

In fact, by similar reasoning \cite{Bremner2010ClassicalSO} (Theorem 2) showed that no classical algorithm can even \textit{weakly} approximately sample from IQP circuits - i.e. up to some fixed multiplicative error. To extend these hardness results to approximate sampling (up to some additive error) in total variation distance, we require stronger assumptions. 

\cite{Bremner2015AveragecaseCV} were the first to show that, assuming an additional complexity-theoretic conjecture on the average-case hardness of computing partition functions, approximately sampling from the output of IQP circuits remains classically intractable even up to small total variation distance. They noted that the output distribution of IQP circuits, 

\begin{equation}
    p_x = |\bra{x}C\ket{0}^{\otimes n}|^2 = 2^{-n}\cdot \big|\mathcal{Z}_x\big|^2,
\end{equation}

\noindent precisely resembles a complex-valued partition function, defined by $w_{u, v}, w_u$ real-valued edge and vertex weights on some underlying architecture graph $G$:
\begin{gather}
    \mathcal{Z}_x = \sum_{z\in \{\pm 1\}^n} \exp\bigg[i\bigg(\sum_{<u, v>}w_{uv}z_uz_v + \sum_u (\pi\cdot x_u + w_u) z_u\bigg].
\end{gather}

 They prove that approximating $\big|\mathcal{Z}_x\big|^2$, and therefore $p_x$, up to multiplicative error is $\# P$ hard in the worst-case, and pose as a conjecture its hardness in the average case over $x$. Under this conjecture, \cite{Bremner2015AveragecaseCV} show that the existence of an efficient classical algorithm to approximately sample from $\{p_x\}$, even up to constant TVD, would imply a collapse of the polynomial hierarchy. 

However, the original results of \cite{Bremner2015AveragecaseCV} referred to a complete graph $G$, which, roughly speaking, correspond to IQP circuits of some polynomial depth. In follow up work by the same authors \cite{Bremner2016AchievingQS}, they reduced the circuit depth to logarithmic under a sparsified version of the graph $G$. It was only in \cite{Gao2016QuantumSF} and \cite{BermejoVega2017ArchitecturesFQ} that the $\# P$ hardness of approximately computing $p_x$ on 2D circuit architectures was established (in the worst-case), corresponding to constant depth IQP circuits in 2D. Their analogous average-case conjecture for approximately computing $p_x$ on 2D circuits, is reproduced below:

\begin{conjecture}
[\cite{Gao2016QuantumSF}]\label{conjecture:mixture_of_errors} There exists a choice of vertex and edge weights $\{w_{uv}, w_u\}_{u, v\in[n]}$ on a 2D lattice $G$, and constants $\epsilon, \delta$, such that approximating the measurement distribution $\{p_x\}$ to the following mixture of multiplicative and additive errors 
\begin{equation}
    |\Tilde{p}_x - p_x| \leq \frac{1}{\poly(n)}\cdot p_x + \frac{\epsilon}{\delta\cdot 2^{n}}
\end{equation}

\noindent is $\#$P hard for any $1-\delta$ fraction of instances $x$.
\end{conjecture}

\cite{Gao2016QuantumSF} show that \cref{conjecture:mixture_of_errors} implies \cref{theorem:sampling}:

\begin{theorem} [\cite{Gao2016QuantumSF}, restatement of \cref{theorem:sampling}]
Assuming \cref{conjecture:mixture_of_errors}, simulating the distribution $\{p_x\}$ up to $\epsilon$ total variation distance is classically intractable, assuming $\mathsf{PH}$ doesn't collapse. 
\end{theorem}

A related result was shown by \cite{BermejoVega2017ArchitecturesFQ}. They start from the (weaker) conjecture that computing $p(x)$ up to a multiplicative factor is hard-on-average, and combine it with a further conjecture on the anti-concentration of the output distribution of random linear-depth IQP circuits. Put together, they also arrive at \cref{theorem:sampling}.

\section{Fault Tolerance of IQP Circuits under Input Noise}
\label{section:fault-tolerance}

We dedicate this section to a proof of \cref{lemma:results-iqp-ft}, on the fault tolerance of IQP circuits under input noise. 

\begin{lemma}\label{lemma:IQP-FT-Distillation}
    Fix an input noise rate $p< \frac{1}{2}$ and a positive integer $D$. Let $C$ be an $n$ qubit IQP circuit with depth $d$ and lightcone size $\ell$. Then, there exists another quantum circuit $\Tilde{C}$, such that a sample from the output of $\Tilde{C}$ under input bit-flip errors can be post-processed using an efficient classical algorithm into a sample $\epsilon$-close to the output distribution of $C$. The circuit $\Tilde{C}$
    \begin{enumerate}
        \item acts on $O(n\log \frac{n}{\epsilon})$ qubits,
        \item has lightcone size $\ell + O\left(D\log^{1/D}\left(\frac{n}{\epsilon}\right)\right)$,
        \item depth $d + O\left(D\log^{1/D} \left(\frac{n}{\epsilon}\right)\right)$, and
        \item the locality of its parent Hamiltonian is $\ell + O(D)$.
    \end{enumerate}
\end{lemma}

For any sufficiently large constant $D$, we recover the claimed fault-tolerance result of \cref{lemma:results-iqp-ft}. If $D = O(\log \log \frac{n}{\epsilon})$, then the circuit depth, lightcone size, and locality are all increased by an additive $O(\log \log \frac{n}{\epsilon})$ factor.

At a high level, our approach is based on pre-processing each of the $n$ input bits into ``code-blocks" or gadgets of size $k = O(\log \frac{n}{\epsilon})$ bits, where each gadget has a designated ``root" bit. The $n$ root bits are then input into the IQP circuit $C$. Since bit-flip errors commute with the IQP circuit, to be able to sample from the original output distribution of $C$, it suffices to \textit{identify} these root bits. Indeed, we emphasize that we do not use the encoding to \textit{correct} the errors within the circuit, as this would require adaptivity and an increase in circuit depth, and instead perform the correction only in post-processing.   

\subsection{The Distillation Gadget}

We place the noisy bits into a tree of arity $B$ (a ``$B$-tree") of depth $D$. For notational convenience, let us partition the nodes in the tree into disjoint subsets, $L_1\cup L_2\cdots \cup L_D = [k]$, the ``layers" of the tree. Moreover, for each node $u$ in the tree, let the subset $N_u$ denote its children or (downwards) neighbors in the tree. The encoding circuit proceeds over the layers from the leaves to the root, where at the $i$th layer $L_i$ of the tree, a CNOT gate is applied from each parent bit to each of its children. 

Note that the size of the tree $k$ is implicitly defined by $B$ and $d$: $k = \sum_{j=0}^{D-1} B^j = \Theta(B^D)$.

\begin{algorithm}[H]
\label{alg:gadget}
    \setstretch{1.35}
    \caption{The Distillation Gadget $U$}
    
    \KwInput{$k$ qubits in the computational basis $\ket{s}$, where $s\leftarrow \text{Ber}^k(p/2)$.}

    \begin{algorithmic}[1]
    
    \State For each layer $i\in [2, \cdots, D]$ from leaves to root, 

    \State For each child $c\in N_p$ of a parent node $p$, apply a $\cnot$ gate from $p$ to $c$.
    \begin{equation*}
        \prod_{i\in [D]} \bigotimes_{p\in L_i} \bigg(\prod_{c\in N_p} \cnot_{p, c}\bigg)\ket{s} \equiv U\ket{s}.
    \end{equation*}
    \end{algorithmic}

\end{algorithm}

We emphasize that the ordering of operations, from leaves to root, matters crucially. In this manner, the $i$th layer acts as a ``parity check syndrome" for the $(i+1)$st. When implemented using 2-qubit gates, the depth of the distillation circuit is $B\cdot D$, as the $\cnot$ gates at the same layer but operating on different subtrees can be performed in parallel, but the $B$ $\cnot$ gates which act on the same parent must be performed sequentially.

\subsection{The Decoding Algorithm}

Next, suppose that all the qubits of $U\ket{s}$ except for that at the root of the tree have been measured, resulting in bits $b_2, \cdots b_k$. Can we reconstruct $s_1$, the bit at the root? The decoding algorithm below traverses the tree layer by layer, from leaves to root, attempting to reconstruct the bit $s_p$ of the next layer. 

\begin{algorithm}[H]
    \setstretch{1.35}
    \label{alg:decoding}
    \caption{The Decoding Algorithm}
    \KwInput{$(k-1)$ bits $b_2, \cdots, b_k$, organized into a $B$-tree, where the root bit has been removed.}
    \KwOutput{A single bit $\Tilde{s}_1$, a guess for the bit at the root.}

    \begin{algorithmic}[1]
    
    \State At the leaves $L_1\subset [k]$, let us denote $\Tilde{b}_u = b_u$ for $u\in L_1$.

    \item For each layer $i\in [2, \cdots, D]$, from leaves to root,

    \State For each parent node $p\in L_i$, let $\Tilde{s}_p = \text{Maj}(\Tilde{b}_c:c\in N_p)$ be the majority of its children bits. 
    
    \State If the root hasn't been reached, update $\Tilde{b}_p \leftarrow \Tilde{s}_p \oplus b_p$. Otherwise, output $\Tilde{s}_1$.
    \end{algorithmic}
\end{algorithm}

The decoding algorithm above maintains the invariant that $\Tilde{s}_u$ is a ``guess" for the original noisy bit $s_u$ input into the distillation gadget. Since $U$ acts from leaves to root, the children in each layer contain (with high probability) the necessary information to reconstruct the parents' bit $s_p$. Together with the measurement outcome $b_p$ - which reveals information about the layer above - we can continue the reconstruction up the tree.

\subsection{Analysis}
 
We divide the analysis into three claims, which consider the correctness, the lightcone size of the circuit, and the ``$Z$-locality" of the distillation gadget which determines the locality of the parent Hamiltonian.

\begin{claim}
    [Correctness]\label{claim:distillation-correctness} Fix any noise rate $p\leq \frac{1-\delta}{2} $ and let $B = \Omega(\delta^{-2})$. Then, the effective bit-flip error rate at the root of the depth $d$ $B$-tree is $\leq 2^{-B^{\Omega(d)}}$.
\end{claim}

\begin{proof}

    We prove inductively that the \textit{effective} bit-flip error rate $p_i$ at the $i$th layer, i.e., 
    \begin{equation}
        p_i \equiv \mathbb{P}_{s\leftarrow \ber^k(p)}[s_u \neq \Tilde{s}_u] \text{ for each node } u\in L_i,
    \end{equation}

    \noindent decays doubly-exponentially with the layer index $i>2$. As the base case, $p_{1} = p$ is the probability of a bit-flip error on the leaves. Suppose $p=\frac{1-\delta}{2}$. Then, after the first layer, the probability the majority vote of the children bits is incorrect is
    \begin{gather}
        p_2 \leq \sum_{j = B/2}^B \binom{B}{j} \big(p_1\big)^j \big(1-p_1\big)^{B-j}  = 2^{-B} \sum_{j = B/2}^B \binom{B}{j} \big(1-\delta\big)^j \big(1+\delta\big)^{B-j} \leq \\
        \leq  \big(1-\delta\big)^{B/2} \cdot \big(1+\delta\big)^{B/2} \leq \big(1-\delta^2\big)^{B/2} \leq \frac{1}{16},
    \end{gather}

    \noindent so long as $B$ is chosen to be $\Omega(\delta^{-2})$. For each layer $i\geq 2$, the effective bit-flip error rate on the $(i+1)$st layer is 
    \begin{equation}
        p_{i+1} \leq \sum_{j = B/2}^B \binom{B}{j} \big(p_i\big)^j \cdot \big(1-p_i\big)^{B-j} \leq 2^B (p_i)^{B/2} \leq p_i^{B/4}.
    \end{equation}

    \noindent In this manner, $p_{i+1} \leq 2^{-(B/4)^i}$ for $i\geq 1$.

\end{proof}

\begin{claim}
    [Circuit lightcone size]\label{claim:distillation-locality} The circuit lightcone size of the distillation scheme is $\leq B\cdot D$. 
\end{claim}

\begin{proof}
The lightcone size of the quantum circuit $U$ is upper bounded by the size of the lightcone of the qubits at the leaves of the tree. Crucially, we claim that if 
\begin{equation}
   u = u_1\rightarrow u_2\rightarrow u_3\cdots \rightarrow u_D = \text{root}
\end{equation}

\noindent denotes the path from a leaf $u\in L_1$ to the root, then only the children of these nodes can be in the lightcone of $u$. Indeed, this is since the $\cnot$ gates in \cref{alg:gadget} are applied layer by layer in increasing order, so the only nodes which are causally connected to $u$ in the circuit are its immediate ascendants or their neighbors. In turn, the size of this set is bounded by $B\cdot D$.

\end{proof}

The last key claim makes reference to the locality of the parent Hamiltonian of the distillation circuit, that is, the size of the support of the operator $U(Z_i\otimes \mathbb{I})U^\dag$, maximized over bits $i$ in the gadget. We thank Joel Rajakumar and James Watson for the observation that the locality of the Hamiltonian is only related to the propagation of Pauli-$Z$ instead of the full circuit lightcone (see also \cite{rajakumar2024gibbs}).

\begin{claim}
    [Parent Hamiltonian Locality]\label{claim:z-locality} The locality of parent Hamiltonian of the distillation circuit is $\leq D$.
\end{claim}

\begin{proof}
    The following two circuit identities describe how Pauli $Z$ operators propagate through $\cnot$ gates. 
    \begin{gather}
        \cnot_{i, j} (Z_i\otimes \mathbb{I})\cnot_{i, j} = Z_i\otimes \mathbb{I} \\
        \cnot_{i, j} (\mathbb{I}\otimes Z_j)\cnot_{i, j} = Z_i\otimes Z_j
    \end{gather}

    Crucially, the locality only increases (or propagates) from the target qubit to the control qubit. Applied to our gadget in \cref{alg:gadget}, we conclude that the qubits in the Z-lightcone of any qubit $i$ in the tree, are precisely the ancestors of $i$. Thus, $|\text{supp}(U(Z_i\otimes \mathbb{I})U)|\leq D$, the depth of the tree.
\end{proof}

We are now in a position to conclude the proof of \cref{lemma:IQP-FT-Distillation}.

\begin{proof}

[of \cref{lemma:IQP-FT-Distillation}] By \cref{claim:distillation-correctness}, if $p\leq \frac{1}{2} (1-\delta)$, then, so long as 

\begin{equation}
    B = \max\bigg(\Theta(\delta^{-2}), \log^{1/D} \big(\frac{n}{\epsilon}\big)\bigg)
\end{equation}

\noindent the probability the decoding algorithm incorrectly outputs the bit at the root of the tree is $\leq \epsilon n^{-1}$. By a union bound, all the gadgets succeed with probability $\geq 1-\epsilon$. Conditioned on this event, the output distribution of $\Tilde{C}$ corrected by the output of the $n$ decoding algorithms is exactly that of $C$, which implies the bound on the TV distance. To conclude, the locality parameters are then implied by \cref{claim:distillation-locality} and \cref{claim:z-locality}
\end{proof}

\section{Quantum Advantage in Gibbs Sampling}
\label{section:together}

We dedicate this section to combining all the aforementioned ingredients and concluding the proof of our main result in \cref{theorem:main}. 

\begin{theorem}[General version of \cref{theorem:main}]
\label{thm:maingeneral}

For any constant inverse-temperature $\beta = \Theta(1)$ and integer $L$, there exists a family of $n$-qubit commuting $O(L)$-local Hamiltonians, such that the $n$-qubit Gibbs state $\rho_\beta$ is both
    \begin{enumerate}
        \item \emph{Rapidly Thermalizing.} It can be prepared within small trace distance by the Davies generator (\cref{equation:Lindbladian}) which has mixing time $e^{O(L\cdot \log^{1/L}(n))}$. In addition, this process can be simulated on a quantum computer in time $n\cdot e^{O(L\cdot \log^{1/L}(n))}$. And yet, 
        \item \emph{Classically Intractable.} Under \cref{conjecture:mixture_of_errors}, there is no polynomial time classical algorithm to sample from the measurement outcome distribution $p(x)=\bra{x}\rho_\beta\ket{x}$ within small constant total variation distance. 
    \end{enumerate}
\end{theorem}

In particular, the choice of a sufficiently large constant $L$ recovers our main result of \cref{theorem:main}. When $L = \log\log n$, we obtain a mixing time of $\polylog (n)$.

\begin{proof}[Proof of \cref{thm:maingeneral}]
To begin our proof, let us fix an inverse-temperature $\beta = \Theta(1)$, and consider the equivalent bit-flip error rate
\begin{equation}
    p = (1+e^{\beta})^{-1} < \frac{1}{2},
\end{equation}
as guaranteed by \cref{lemma:gibbs_input_noise}. 

\paragraph{Classical Intractability.} Consider the family of constant-depth, classically intractable, $n$-qubit IQP circuits $C$ guaranteed by \cref{theorem:sampling} (\cref{conjecture:mixture_of_errors}). Using \cref{lemma:results-iqp-ft}, let us fix a depth parameter $L$, and embed each circuit in said family into a new circuit $\Tilde{C}$, which is fault tolerant to input noise of rate $p = \frac{1}{2}(1-\Theta(1))$. $\Tilde{C}$ now has $Z$-locality $O(L)$, circuit depth and lightcone size $O(L\log^{1/L} (\frac{n}{\epsilon}))$; and a noisy sample from $\Tilde{C}$ can be efficiently classically post-processed into a sample $\epsilon$-close in trace distance to an ideal sample from $C$.

Now, consider the family of parent Hamiltonians defined by the family of Fault-Tolerant circuits $\Tilde{C}$, 
\begin{equation}
    H = \sum_i \Tilde{C}\bigg(Z_i\otimes \mathbb{I}_{[n]\setminus i}\bigg)\Tilde{C}^\dagger.
\end{equation}

\noindent The support size of each term is given by the $Z$-locality of the fault-tolerant circuit $\Tilde{C}$, which is $O(L)$.

If, by assumption, there was a polynomial time classical algorithm $\mathcal{A}$ to sample from the Gibbs state of $H$ at inverse-temperature $\beta$, then we could construct a polynomial time classical algorithm to sample from a distribution $\epsilon$-close to the ideal distribution of $C$, as follows: First, construct $\Tilde{C}$ and thus the local terms of $H$ from $C$. Then, leverage $\mathcal{A}$ to sample from $\propto e^{-\beta H}$. Finally, process the output sample using the post-processing algorithm from the fault-tolerance statement of \cref{lemma:results-iqp-ft}.

\paragraph{Rapid Thermalization.} To conclude, via \cref{lemma:results-gibbs-prep}, the Gibbs state of $H$ can be prepared using the Davies generator of \cref{equation:Lindbladian} of mixing time exponential in the circuit lightcone size, $\log n\cdot \exp(O(L\cdot \log^{1/L}(n))) = \exp(O(L\cdot \log^{1/L}(n)))$. To simulate this process on a quantum computer, the overall runtime $n\cdot \exp(O(L\cdot \log^{1/L}(n))) $ has an additional quasi-linear overhead. 
\end{proof}

\begin{remark}\label{remark:order-of-quantifiers}
    \cref{theorem:main} asserts that for every constant temperature, there exists a Hamiltonian $H$ which is classically hard-to-sample from. Conversely, results by \cite{yin2023polynomialtime} and \cite{bakshi2024hightemperature} show that every local Hamiltonian (of fixed degree) has a critical temperature, such that above said threshold one can efficiently classically sample from their Gibbs state. The resolution to this apparent contradiction lies in the order of quantifiers. The degree/locality of our Hamiltonians increases with the temperature, see \cref{section:fault-tolerance} for their dependence on the noise rate. 
\end{remark}

\begin{remark}\label{remark:gate-obfuscation}
    Since the Gibbs state is determined by a low depth quantum circuit $C$, with access to a description of $C$, one could trivially produce it on a quantum computer. However, if given access only to the local Hamiltonian terms $\{h_i\}_i = \{-CZ_iC^\dagger\}_i$, we do not believe it to be computationally efficient to recover the global structure of $C$, in general. While this is not a rigorous statement, we only know how to do so for 1D circuits, via dynamic programming. It is worthwhile to contrast this to the Feynman-Kitaev circuit-to-Hamiltonian mapping \cite{kitaev02}, wherein the gates of the circuit can be exactly read-off from the local Hamiltonian interactions. 
\end{remark}

\section{BQP Completeness with Adaptive Single-Qubit Measurements}
\label{section:completeness}

We dedicate this section to a proof of \cref{theorem:results-mbqc}, on the BQP completeness of Gibbs Sampling with adaptive measurements.

\begin{theorem}\label{theorem:mbqc}
    Fix an inverse-temperature $\beta = \Theta(1)$. Then, there exists an $n$-qubit $O(1)$-local Hamiltonian, whose Gibbs state at inverse-temperature $\beta$ is a universal resource state for quantum computation and is efficiently preparable on a quantum computer. 
\end{theorem}

This result is all but a corollary of our fault tolerance techniques for IQP circuits, applied to measurement-based quantum computation. Indeed, it is well known that 2D cluster states, in addition to single-qubit measurements in adaptively chosen basis on the $X-Y$ plane, is universal for quantum computation. The following lemma shows that one can produce said cluster state out of the Gibbs state of a local Hamiltonian, so long as we are allowed to measure a subset of the qubits, and subsequently apply a Pauli correction to ``distill" out the cluster state.

\begin{lemma}
    There exists a $n$-qubit, $O(1)$-local commuting Hamiltonian, whose Gibbs state at inverse-temperature $\beta$ can be used to prepare a cluster state. That is, by measuring a subset of the qubits of the Gibbs state, and then with 1 round of adaptive Pauli correction, one can produce a 2D cluster state on $O(n/\log \frac{n}{\epsilon})$ qubits with probability $1-\epsilon$. 
\end{lemma}

\begin{proof}
    Let $C$ be the circuit which prepares a 2D cluster state on $m$ qubits, comprised of Hadamard gates and CZ gates. Let $\Tilde{C}$ be the $n=\Theta(m\log \frac{m}{\epsilon})$ qubit circuit defined by the fault tolerance scheme of \cref{lemma:IQP-FT-Distillation}, which is robust to input errors of finite probability $<\frac{1}{2}$. Then, consider the parent Hamiltonian $H$ of $\Tilde{C}$, on $n$ qubits and with locality $O(1)$. 
    
    By construction, its Gibbs state is a quantum-classical state, of classical bits lying in the fault-tolerance gadget of \cref{lemma:IQP-FT-Distillation}, and qubits comprising a cluster-state under input noise. Again, recall that input bit-flip errors are equivalent to output $Z$ errors, due to the gate structure of $C$. From \cref{lemma:IQP-FT-Distillation}, by measuring the classical bits of the fault-tolerance gadget, one can recover the output $Z$ error with probability $1-\epsilon$.
\end{proof}

We remark that the adaptively chosen $X-Y$ measurements can be performed simultaneously with the Pauli corrections. In this manner, after producing the desired resource Gibbs state, it suffices to perform adaptively chosen single-qubit measurements to achieve universal measurement based quantum computation.

\section{Addressing Output Measurement Errors}
\label{sec:measurementnoise}

In this section, we prove \cref{theorem:results-measurement-errors} on sampling from finite-temperature Gibbs states subject to measurement errors. 

\begin{lemma}\label{lemma:measurement-errors}
    Fix an inverse temperature $\beta = \Theta(1)$, and a measurement error rate $p < \frac{1}{2}$. There exists a family of $n$-qubit, $O(\log n)$-local Hamiltonians, such that sampling from their Gibbs state at inverse-temperature $\beta$, under measurement errors of rate $p$, is classically intractable under \cref{theorem:sampling}. Moreover, there exists a $\poly(n)$ time quantum algorithm to produce said Gibbs state.
\end{lemma}

Our construction of \cref{lemma:measurement-errors} is similarly based on the parent Hamiltonians of fault-tolerant IQP circuits, which are hard-to-sample from in the ideal case. We note that the distribution defined by sampling from the Gibbs state of the parent Hamiltonian of a quantum circuit $C$, given measurement errors, corresponds exactly to sampling from $C$ under both input and output noise, albeit with different noise rates. Unfortunately, to address this mixed noise model, we do need to appropriately modify our fault-tolerance scheme. For this purpose, we appeal to prior work by \cite{Bremner2016AchievingQS}, at the cost of a higher locality.

\subsection{Overview}

To model the noise in this section, recall the definition of the bit-flip error channel $\mathcal{D}_p$ in \cref{equation:bit-flip}. Given a quantum circuit $C$ on $n$ qubits, and fixed noise rates $p_{in}, p_{out}\in [0, \frac{1}{2})$, the noisy output distribution of $C$ given input and output noise is given by

\begin{equation}
     p_{C, p_{in}, p_{out}}(x) = \Tr[\ketbra{x}\cdot \mathcal{D}_{p_{out}}^{\otimes n}\circ C\bigg(\mathcal{D}_{p_{in}}\circ (\ketbra{0})\bigg)^{\otimes n} C^\dagger]
\end{equation}

If $\mathcal{A}:\{0, 1\}^{n}\rightarrow \{0, 1\}^{n'}$ is a deterministic classical post-processing algorithm, we denote as $\mathcal{A}\circ p$ the distribution given by sampling $x\leftarrow p$ and outputting $\mathcal{A}(x)$. The following lemma is a fault-tolerance statement for IQP circuits against this input/output noise model.

\begin{lemma}\label{lemma:iqp-ft}
    Let $C$ be an $n$ qubit IQP circuit of depth $d$ and lightcone size $\ell$, and fix input and output bit-flip error rates $p_{in}, p_{out}\in [0, \frac{1}{2})$. Then, for every $r\in \mathbb{N}$ there exists a quantum circuit $C_r$ and a deterministic, $O(n_r)$-time decoding algorithm $\mathcal{A}_r:\{0, 1\}^{n_r}\rightarrow \{0, 1\}^n$, such that in the presence of input and output noise, the statistical distance  
        \begin{equation}
            \|\mathcal{A}_r\circ p_{C_r, p_{in}, p_{out}} - p_{C, 0, 0}\|_1\leq n\cdot (4q(1-q))^{r/2}, \text{ where } q = p_{in}(1-p_{out}) + p_{out}(1-p_{in}) < \frac{1}{2}.
        \end{equation}

        \noindent Moreover, $C_r$ is defined on $n_r=n\cdot r$ qubits, has depth $d_r = d\cdot \log r$ and lightcone size $\leq \ell\cdot r$. 
\end{lemma}

In other words, noisy samples from $C_r$ can be post-processed into nearly-ideal samples from $C$. Note that $q<\frac{1}{2}$ implies the total variation distance above decays exponentially with $r$.

\begin{corollary}
    Fix input and output bit-flip error rates $< \frac{1}{2}$. Then, any IQP circuit on $n$ qubits and constant depth can be efficiently transformed into a quantum circuit of $O(\log \log n)$ depth and $O(\log n)$ lightcone size, robust to input and output noise with error $n^{-\Omega(1)}$.
\end{corollary}

Starting from the hard-to-sample IQP circuits ensured by \cref{theorem:sampling}, we can construct circuits fault-tolerant to input and output noise via the Corollary above. In turn, these fault-tolerant circuits define a parent Hamiltonian, which is rapidly thermalizing (via \cref{lemma:results-gibbs-prep}), and yet, classically hard to sample from. Put together, we prove \cref{lemma:measurement-errors}.

\subsection{Analysis}

We remark that if the circuit $C$ itself is an IQP circuit, then the bit-flip noise model $\mathcal{B}_p$ commutes with the circuit, and thus the input/output noise is equivalent to input noise at a higher rate: $p_{C, p_{in}, p_{out}}(x) = p_{C, q, 0}(x)$, with 

\begin{equation}
    q = p_{in}(1-p_{out}) + p_{out}(1-p_{in})<\frac{1}{2}
\end{equation}

\noindent To leverage this equivalence, however, we need to design a fault-tolerant circuit which itself is an IQP circuit. Fortunately, here we can appeal to \cite{Bremner2016AchievingQS}, who achieved precisely that. To summarize their construction, their circuit embedding leverages the following property of IQP circuits. The diagonal part $D$ of any IQP circuit can be expressed as a matrix-exponential of a polynomial of $Z$ Pauli matrices:
\begin{equation}
    D = \exp\bigg[i\sum_{{j\in [m]}} \theta_j \bigotimes_{i\in [n]} Z_i^{M_{ji}}\bigg], \text{ for real coefficients }\{\theta_j\}, \text{ and a boolean matrix } M\in \mathbb{F}_2^{m\times n}.
\end{equation}

If $D$ is comprised of 2-qubit gates, then the weight of any row of $M$ is $\leq 2$. Now, suppose $G\in \mathbb{F}_2^{(n\cdot r)\times n}$ is the generator matrix of a repetition code, on $n' = n\cdot r$ bits and rate $n/n' = \frac{1}{r}$. \cite{Bremner2016AchievingQS} observe that the new IQP circuit defined by mapping $M\rightarrow \Tilde{M}=M\cdot G^T$ is robust to input noise, up to (roughly) the random-error-correction capacity of $G$. Indeed, this follows from the fact that 
\begin{equation}
    \bra{G^T x}D \ket{G^Tx} = \bra{x}\Tilde{D}\ket{x}, \forall x\in\{0, 1\}^{n'}.
\end{equation}

Therefore, the output distribution of the new circuit $\Tilde{C}$ under input (or output) noise is the same as sampling $y\in \{0, 1\}^n$ from $C$, encoding $y$ into the code $\Tilde{y} = Gy\in \{0, 1\}^{n'}$, and finally flipping each entry of $\Tilde{y}$ independently with probability $q$. If the repetition code can tolerate random bit-flip errors with rate $q$, then one can approximately sample from $C$ using noisy samples from $\Tilde{C}$.

The caveat in their approach is that the resulting IQP circuits maybe polynomially larger. Indeed, each two qubit gate in the original circuit $C$, is mapped to a $2\cdot r$ multi-qubit gate in $\Tilde{C}$:
\begin{equation}
    e^{i\theta Z_a\otimes Z_b}\rightarrow e^{i\theta Z_a^1\otimes Z_a^2\cdots Z_a^r\otimes Z_b^1\cdots Z_b^r}
\end{equation}

\noindent which is complex to implement using only diagonal operations. Instead, we dispense with the requirement that the intermediate gates in the circuit be diagonal (and thus the circuit is not an IQP circuit), however, globally it is equivalent to the same (IQP) unitary operation. 

\begin{definition}
    A \emph{$k$-local Pauli rotation gate} is the $k$ qubit unitary $U$ defined by an angle $\theta\in [0, 2\pi]$ and a $k$-qubit Pauli $P$ where $U = e^{i\theta P}$.
\end{definition}

Of particular note to us are multi-controlled $Z$ rotations, where $P = Z_1\otimes Z_2\cdots Z_k$.

\begin{claim}
    Any $k$-local Pauli rotation gate can be implemented using an $\leq \log k$ depth circuit on a fully connected architecture of 2-qubit gates. 
\end{claim}

For simplicity, we prove the above for multi-qubit $Z$ Paulis, as the general case is analogous. 

\begin{proof}
    Let $U$ be a $k$-local Z rotation gate, and $V$ be any unitary. Then, the identity $Ve^{i\theta P}V^\dagger = e^{i\theta VPV^\dagger}$ tells us that it suffices to find a depth $d\leq \log k$ Clifford circuit $V$ such that $V(\otimes_i^k Z_i)V^\dagger= Z_1\otimes \mathbb{I}_{[k]\setminus 1}$. We claim that this can be done recursively, where each layer of $V$ halves the weight of the remaining Z's. Indeed, since $(\mathbb{I}\otimes Z) = \cnot (Z\otimes Z)\cnot^\dagger$, layers of $\cnot$ gates on a matching of the remaining $Z$'s will suffice. 
\end{proof}

To prove our statement, we instantiate the lemma below with our implementation of multi-controlled $Z$ gates.

\begin{lemma}[\cite{Bremner2016AchievingQS}\label{lemma:iqp-ft-multi-controlled}]
    Let $C$ be an $n$ qubit IQP circuit of depth $d$. Then, for every $r\in \mathbb{N}$, there exists a deterministic, $O(n\cdot r)$-time decoding algorithm $\mathcal{A}_r:\{0, 1\}^{n\cdot r}\rightarrow \{0, 1\}^n$, and a quantum circuit $C_r$ on $n_r=n\cdot r$ qubits, comprised only of Hadamard gates and $O(d)$ layers of $\leq 2r$-local Z rotation gates, satisfying
    \begin{enumerate}
        \item In the absence of noise, the distribution $\mathcal{A}_r\circ p_{C_r, 0, 0}$ given by sampling $y\leftarrow p_{\Tilde{C}_r, 0, 0}$ from the output of $C_r$, and outputting $\mathcal{A}_r(y)$, is the same as sampling from $C$. 
        \item In the presence of input-level noise with probability $q$, the statistical distance  
        \begin{equation}
            \|\mathcal{A}_r\circ p_{\Tilde{C}_r, q, 0} - p_{C, 0, 0}\|_1\leq n\cdot (4\cdot q\cdot (1-q))^{r/2}.
        \end{equation}
    \end{enumerate}
\end{lemma}

\section*{Acknowledgements}
We thank Anurag Anshu, Ángela Capel, Lijie Chen, Daniel Stilck França, Sevag Gharibian and Umesh Vazirani for helpful discussions. This work was done in part while the authors were visiting the Simons Institute for the Theory of Computing. 
T.B.~acknowledges support by the National Science Foundation under Grant No. DGE 2146752.
Y.L.~is supported by DOE Grant No. DE-SC0024124, NSF Grant No. 2311733, and DOE Quantum Systems Accelerator.

\printbibliography

\end{document}